\documentclass[a4paper]{article}
\usepackage[a4paper, total={5.5in, 9in}]{geometry}
\usepackage{authblk}
\usepackage{bm}
\usepackage{float}
\usepackage[T1]{fontenc}
\usepackage{changepage}
\usepackage[style=default]{caption}
\usepackage[noEnd=true]{algpseudocodex}
\usepackage{parskip}
\usepackage{booktabs}
\usepackage[dvipsnames]{xcolor}
\usepackage[ruled]{algorithm}
\usepackage{graphicx}
\usepackage[shortlabels]{enumitem}
\usepackage{fontawesome}
\usepackage{multirow}
\usepackage{microtype}
\usepackage{wrapfig}
\usepackage{hyperref}
\usepackage{stmaryrd}
\hypersetup{colorlinks=true, citecolor=teal, linkcolor=teal, urlcolor=teal}
\usepackage[osf,sc]{mathpazo}
\usepackage[scaled=0.85]{helvet}
\usepackage{array}
\usepackage{tikzscale}
\usepackage{tikz, tikz-cd}
\usetikzlibrary{positioning}
\usetikzlibrary{chains}
\usetikzlibrary{decorations.pathreplacing}
\usetikzlibrary{calc}
\usepackage[kw]{pseudo}
\pseudoset{
    compact,
    label=\scriptsize\sffamily\color{gray}\arabic*,
    ctfont=\sffamily\color{gray},
    ct-left=\ctfont{\# },
    ct-right=,
    line-height=1.2,
    indent-length=3ex,
    left-margin=0pt,
}

\setlist[description]{parsep=0.1px}

\usepackage[shortlabels]{enumitem}

\usepackage{amsmath,amsfonts,amssymb,amsthm}
\usepackage{cleveref}

\newtheorem{definition}{Definition}
\newtheorem{theorem}{Theorem}

\newtheorem{corollary}{Corollary}
\newtheorem{proposition}{Proposition}
\newtheorem{lemma}{Lemma}
\newtheorem{notation}{Notation}
\theoremstyle{remark}
\newtheorem{remark}{Remark}
\newtheorem{example}{Example}

\usepackage[
style=alphabetic,
backend=biber,
isbn=false,
alldates=year,
sortcites=true,
doi=true,
eprint=false
]{biblatex}
\AtEveryBibitem{
    \clearfield{urlyear}
    \clearfield{urlmonth}
}
\bibliography{./library.bib}

\renewbibmacro*{doi+url}{%
  \iftoggle{bbx:doi}
  {\printfield{doi}%
    \iffieldundef{doi}{}{}}
  {}%
  \newunit\newblock
  \iftoggle{bbx:eprint}
  {\usebibmacro{eprint}%
    \iffieldundef{eprint}{}{}}
  {}%
  \newunit\newblock
  \iftoggle{bbx:url}
  {\usebibmacro{url+urldate}%
    \iffieldundef{url}{}{}}
  {}}

\AtEveryBibitem{\clearfield{month}}
\AtEveryBibitem{\clearfield{day}}
\DeclareNameAlias{sortname}{first-last}



\newcommand{\set}[1]{\left\{#1\right\}}
\newcommand{\setdes}[2]{\left\{#1\;\middle|\;#2\right\}}

\def\Var{\mathbf{V}}

\def\NN{\mathbb{N}}
\def\RR{\mathbb{R}}
\def\CC{\mathbb{C}}
\def\ZZ{\mathbb{Z}}
\def\FF{\mathbb{F}}
\def\DN{\mathbb{D}}
\def\PP{\mathbb{P}}

\newcommand{\irng}[1]{\llbracket #1 \rrbracket}

\def\xx{\mathbf{x}}
\def\ww{\mathbf{w}}

\def\corn{\delta}
\def\tvar{\mathcal{T}}

\DeclareMathOperator{\lt}{lt}

\DeclareMathOperator{\supp}{supp}

\DeclareMathOperator{\trop}{trop}
\DeclareMathOperator{\mv}{MV}
\DeclareMathOperator{\vol}{vol}
\DeclareMathOperator{\svol}{svol}
\DeclareMathOperator{\Conv}{conv}

\DeclareMathOperator{\codim}{codim}

\DeclareMathOperator{\mult}{mult}
\DeclareMathOperator{\canc}{canc}

\newcommand{\sgn}{\mathop{\rm sgn}\nolimits}
\def\melim{m_{\text{elim}}}

\usepackage{xcolor}

\title{Engineered Complete Intersections:\\ Algorithmic Aspects}

\author{Alexander Esterov \thanks{\href{aes@lims.ac.uk}{aes@lims.ac.uk}, \href{alexander.esterov@gmail.com}{alexander.esterov@gmail.com}, London Institute for Mathematical Sciences, London, UK}
  \and Rafael Mohr  \thanks{\href{mailto:rafaeldavid.mohr@kuleuven.be}{rafaeldavid.mohr@kuleuven.be}, Department of Computer Science, KU Leuven, Belgium}
  \and Yulia Mukhina \thanks{\href{mailto:yulia.mukhina@lix.polytechnique.fr}{yulia.mukhina@lix.polytechnique.fr}, LIX, CNRS, École Polytechnique, Palaiseau, France}}

\begin{document}

\maketitle

\begin{abstract}
  \textit{Engineered Complete Intersections} (ECI's) are a class of
  sparse polynomial systems frequently arising in a number of
  contexts, both in pure mathematics (e.g. enumerative geometry) and
  applications (e.g. chemical reaction networks).

  Based on the theoretical results given in
  \cite{esterovEngineeredCompleteIntersections2024,
    esterovEngineeredCompleteIntersections2025}, we give several
  contributions. First we give a new effective technique to
  tropicalize such systems by generalizing the classical notion of
  mixed subdivisions \cite{huber1995} to ECI's with the particular
  goal to efficiently count solutions of square systems of equations
  in ECI form. We further design a \textit{tropical homotopy
    continuation} algorithm for computing such mixed subdivisions,
  inspired by \cite{jensenTropicalHomotopyContinuation2016,
    malajovichComputingMixedVolume2017a,
    daiseyFrameworkGeneralizedTropical2024}. Our techniques can be
  used to numerically solve such systems by coupling them with the
  algorithms introduced in
  \cite{helminckTropicalMethodSolving2024}. Finally, we give an
  algorithm to compute Newton polytopes of eliminants of ECI's. This
  gives a new way to compute, for example, Newton polytopes of
  so-called $A$-discriminants. Coupled with evaluation-interpolation
  paradigms our algorithm gives an efficient approach to compute such
  eliminants.

  We implemented our algorithms in the form of a software package
  which we use to demonstrate their practical feasibility on a range of
  examples.
\end{abstract}

\section{Introduction}
\label{sec:intro}

\paragraph*{Engineered Complete Intersections.}

An \textit{engineered complete intersection} (ECI) is a class of
polynomial systems constructed as follows: Let $A \subset \ZZ^n$ be a finite
multiset and let $V \subset \CC^{k\times A}$ be a matrix with
$k \leq n$ whose columns are indexed by $A$ and whose rank is equal to
$k$. Then the associated ECI is the parametric Laurent
polynomial system
\[f_i:=\sum_{a\in A}c_av_{i,a}\xx^a = 0;\;i=1,\hdots,k,\] where the
$c_a$, $a\in A$, are symbolic coefficients and $\xx$ is a set of $n$ variables.

In words, such an ECI arises by pluging the monomial functions with
indeterminate coefficients $c_a\xx^a$ into the system of linear
equations with matrix $V$. In more invariant terms, it is the
intersection of a general coset of a given subtorus in the algebraic
torus $(\CC^*)^A$ with a given vector subspace of $\CC^A$.

ECI's encode parametric polynomial systems where each equation has
coefficients that linearly depend on the remaining equations in a
fixed way for generic choices of coefficients. For this reason,
0-dimensional ECI's are often referred to as {\em vertically
  parametrized systems} in the literature.

Prominent examples of ECI's are given by systems of critical or
singular points: For example, fixing a set $A\subset \ZZ^n$ and $n$
variables $\xx := \set{x_1,\ldots, x_n}$ as well as
\[f_A := \sum_{a\in A} c_a\xx^a,\] the parametric system whose solutions are the
critical points of the projection of the vanishing locus of $f_A$ in
the torus $(\CC^{*})^n$ to the $x_1$-axis is an ECI, as this class of
systems is given by
\[f_A = x_2\partial_{x_2}f_{A} = \ldots = x_n\partial_{x_n}f_A = 0.\] Notably, ECI's also appear
as polynomial systems of equations encoding steady states of dynamical
systems coming from so-called {\em chemical reaction networks}, see
e.g. \cite{conradiIdentifyingParameterRegions2017,
  dickensteinAlgebraicGeometryTools2020}.

There are many other examples of algebraic varieties and systems of
equations of special interest arising from an ECI, see
e.g. \cite{esterovEngineeredCompleteIntersections2024,
  esterovEngineeredCompleteIntersections2025,
  selyaninNewtonNumbersVanishing2025}.

\paragraph*{Contributions and Related Work.}

Engineered complete intersections are studied from a theoretical point
of view by the first author of this paper in
\cite{esterovEngineeredCompleteIntersections2024,
  esterovEngineeredCompleteIntersections2025}. More precisely,
\cite{esterovEngineeredCompleteIntersections2024,
  esterovEngineeredCompleteIntersections2025} studied ECI's from the
point of view of \textit{tropical geometry} (see
e.g. \cite{maclaganIntroductionTropicalGeometry2021,
  mikhalkinTropicalGeometryIts2006} for an introduction to this
subject). Our contribution is essentially a set of algorithms that can
be used to compute tropical data associated to ECI's.

While tropicalizations of polynomial ideals are in general difficult
to compute (see e.g. \cite{bogartComputingTropicalVarieties2007}),
there are more favorable cases: When $F$ is a \textit{generic} system
of Laurent polynomials with fixed Newton polytopes, its
tropicalization depends only on those Newton polytopes, see
e.g. \cite{sturmfelsEliminationTheoryTropical,
  fultonIntersectionTheoryToric1997}. When $F$ has as many equations
as variables, and thus finitely many solutions, and its
  coefficients are power series of a parameter (or specialized to any
  other valued field) its tropicalization is given by a certain
subdivision of the involved Newton polytopes, called a \textit{mixed
  subdivision}, and encodes in particular the number of solutions of
$F$, given as the {\em mixed volume} of the involved Newton polytopes
by the BKK theorem \cite{bernshtein1975, kouchnirenko1976,
  khovanskii1995}. These subdivisions, together with an algorithm to
compute them, were introduced in the celebrated work \cite{huber1995}
and can be used in particular to numerically
solve such a system $F$ using \textit{homotopy methods}
\cite{sommese2005} in an optimized way. Note that such
systems $F$ are particular instances of ECI's.

In \cite{esterovEngineeredCompleteIntersections2024}, it was shown
that the polyhedral framework established by \cite{bernshtein1975,
  kouchnirenko1976, khovanskii1995} could be generalized to ECI's by
replacing the implied Newton polytopes with \textit{conewise linear
  functions} (depending only on $A$ and $V$, and actually moreover
only on the matroid defined by the columns of $V$). The
tropicalization of a generic instance of the ECI given by $V$ and $A$
then depends only on these conewise linear functions and can be
constructed using the theory established in
\cite{esterovTropicalVarietiesPolynomial2012}. This generalizes the
setting above in the sense that the \textit{support function} of a
polytope is a conewise linear function and, for a generic Laurent
polynomial system with fixed Newton polytopes, the results of
\cite{esterovEngineeredCompleteIntersections2024} then specialize to
the results mentioned above.

When the ECI given by $V$ and $A\subset \ZZ^n$ is zero-dimensional (in the
sense that $V$ consists of $n$ linearly independent rows) then the
root count of a generic instantiation of this ECI coincides with the
{\em mixed volume} of the associated conewise linear functions, a
quantity introduced in \cite{esterovTropicalVarietiesPolynomial2012}.

Our first contribution is to transport the notion of a mixed
subdivision established in \cite{huber1995} for collections of
polytopes to ECI's. We furthermore give an algorithm to compute such
mixed subdivisions of ECI's. Importantly, these mixed subdivisions can
again be used to numerically solve generic instances of
$0$-dimensional ECI's in an optimized way, using the techniques
introduced in \cite{helminckTropicalMethodSolving2024}, and the sum
of the volumes of all the polyhedra contained in this subdivision
coincides with mixed volume of the underlying ECI.

Our algorithm is a tropical version of the homotopy techniques used in
numerical algebraic geometry mentioned above. Such \textit{tropical
  homotopy} algorithms were proposed in
\cite{jensenTropicalHomotopyContinuation2016,
  malajovichComputingMixedVolume2017a} to compute mixed subdivisions
of polytopes, our algorithm can therefore be seen as a generalization
of these techniques to ECI's. As the algorithms given in
\cite{jensenTropicalHomotopyContinuation2016,
  malajovichComputingMixedVolume2017a}, our algorithm works
essentially by tracking a mixed subdivision along a linear path through
an ambient space of dimension equal to the size of $A$. Each time
this path crosses a facet of a so-called {\em mixed cell cone}, the
mixed subdivision is modified computationally.

In \cite{helminckTropicalMethodSolving2024}, another effective
technique to compute the tropicalization of an ECI was proposed: The
authors showed that it is given by the projection of the
\textit{stable intersection} of a tropical linear space with a
collection of tropicalizations of hypersurfaces defined by
binomials. The authors note that a direct implementation of the
computation of the tropicalization of an ECI may be challenging, as
the tropical linear space to be computed may be larger than the final
result. Another result based on tropicalizing linear spaces was
given recently in \cite{feliuRootBoundsVertical2026} which in
particular yields an algorithm to compute root counts of generic
instances of ECI's. This work also gives various bounds on the number
of real solutions of instances of ECI's.

In order to circumvent the direct computation of a tropical linear
space, a tropical homotopy algorithm is given in
\cite{daiseyFrameworkGeneralizedTropical2024,
  daiseyTropicalHomotopyContinuation2025} which is able to treat
stable intersections of tropical linear spaces and hypersurfaces and
thus extends beyond the case of ECI's. This homotopy algorithm relies
theoretically on subdivisions of so-called {\em matroid polytopes}
which become rapidly infeasible to handle computationally. As
mentioned in \cite{daiseyTropicalHomotopyContinuation2025}, this has
the consequence that the mixed cell cones mentioned above cannot be
computed with directly in the context of this technique and that
instead one has to work computationally with {\em chains of flats} of
matroids which provide a certain encoding of tropical linear spaces
\cite[see e.g. Theorem 4.2.6
in][]{maclaganIntroductionTropicalGeometry2021}. We note that no
experimental evaluation of the resulting homotopy algorithm is given
in \cite{daiseyFrameworkGeneralizedTropical2024,
  daiseyTropicalHomotopyContinuation2025}.

As mentioned, our algorithm instead relies on the conewise linear
functions associated to an ECI constructed in
\cite{esterovEngineeredCompleteIntersections2024,
  esterovEngineeredCompleteIntersections2025}. Each constituent
equation $f_i$ of the given ECI gives one such conewise linear
function and is constructed essentially by cancelling terms in $f_i$
using \textit{scalar} multiples of the equations
$f_1,\ldots, f_{i-1}$. This methodological core results in an algorithm
which is able to manipulate the involved mixed cell cones directly and
thus gives a relatively simple procedure to change the considered
mixed subdivision when a facet of a mixed cell cone is crossed, close
to the one given in \cite{jensenTropicalHomotopyContinuation2016}. We
evaluate our algorithm experimentally, showing that it is able to
handle many ECI's of practical interest with relative ease, for
example ECI's coming from the theory of chemical reaction networks.

Furthermore, a much more general class of polynomial systems
(called \textit{nondegenerate upon cancellation}, see Section 3 in
\cite{esterovEngineeredCompleteIntersections2025}) can be tropicalized
similarly by cancelling terms in $f_i$ using \textit{monomial}
multiples of the equations $f_1,\ldots, f_{i-1}$. Our algorithmic techniques
thus potentially generalize to this setting as well.

Our algorithm also allows us to make effective the {\em real
  patchworking} techniques for ECI's introduced in
\cite{esterovEngineeredCompleteIntersections2025} which extend the
classical tropical patchworking techniques given in
\cite{sturmfelsVirosTheoremComplete1994, viro1982gluing}. Combined
with out algorithm, these techniques allow one to algorithmically find
the number of real solutions of an ECI depending on a single parameter
for a sufficiently large parameter value. This combination requires
one to count the number of real solutions of certain polynomial
systems whose supports are vertices of mutually transversal
simplices. We show how to do this algorithmically by reducing to
solving a linear system over the finite field with two elements. As
an application, we construct algorithmically a degree 4 hypersurface
in 3 variables such that all cusp singularities of its discriminant curve
are real, see \Cref{fig:cusps}.

\begin{figure}[h]
  \label{fig:cusps}
\caption{This is the discriminant curve of the projection of a torus into the plane ${\mathbb R}_{y,z}^2$, with cusp and self-intersection singular points in blue and red. If the torus is given by a general polynomial equation $f(x,y,z)=0$, the cusps are defined by the ECI $f=f_x=f_{xx}=0$.}
\centering
\includegraphics[width=5cm]{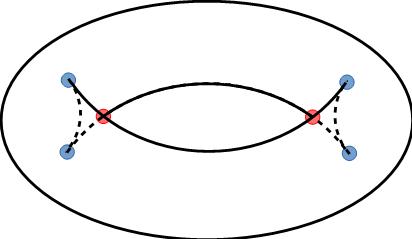}
\end{figure}

Our second contribution concerns the following problem: Given an ECI
$(V, A)$ with $A\subset \ZZ^n\times \ZZ^{k}$ for integers $k$ and $n$ and
$V\in \CC^{(n+1) \times A}$ of rank $n + 1$, the projection of the algebraic
set defined by a generic instance of $(V, A)$ to $(\CC^{*})^k$ is a
hypersurface $\Var(g)$ which one may wish to compute. Examples of
particular cases of this problem are given by the computation of
certain discriminants, in particular so-called {\em $A$-discriminants}
\cite{gelfand1994}, the computation of critical values of projection
maps \cite[see e.g. the introduction
of][]{esterovEngineeredCompleteIntersections2025} and the computation
of {\em minimal polynomials} of dynamical systems with polynomial
right hand sides \cite{mukhina2025}.

The Newton polytope of $g$ is independent of the chosen generic
specialization of $(V, A)$, depending only on the tropicalization of
$(V, A)$, and describing this Newton polytope in terms of the conewise
linear functions associated to $(V, A)$ was done from a theoretical
point of view in \cite{esterovEngineeredCompleteIntersections2025}.
We give here an algorithm to compute the {\em Newton polytope} of $g$
from the data $(V, A)$, based on this theoretical description. Knowing
the Newton polytope of $g$ can be used to compute $g$ itself, for example
with an evaluation-interpolation approach, or by avoiding a direct elimination
computation using Gröbner bases with an approach similar to the so-called
FGLM algorithm \cite{faugere1993a}.

When the tropicalization of an algebraic set is known, and its
eliminant is a hypersurface, general techniques exist to compute the
Newton polytope of this hypersurface from the given tropicalization,
see e.g. \cite{sturmfelsEliminationTheoryTropical,
  roseTropicalImplicitizationRevisited2025}. In the particular case of
polynomial systems with fixed Newton polytopes and generic
coefficients, the sought Newton polytope is the so-called {\em mixed
  fiber polytope} of the input Newton polytopes \cite{esterov2008,
  mcmullenMixedFibrePolytopes2004}. A dedicated algorithm for
computing mixed fiber polytopes was given in
\cite{mohrComputationNewtonPolytopes2025} and therein shown to be much
faster than the more general purpose tropical elimination techniques
mentioned above in the case considered. This algorithm worked by
giving an effective formula which, from a certain mixed subdivision,
extracts a vertex for the sought Newton polytope. Using this formula
one can avoid the potentially costly computation of the complete
tropicalization of the input system. The algorithm we give works along
a similar formula which we prove for the more general case of ECI's.
We evaluate our algorithm experimentally, showing that it is able to
compute Newton polytopes of $A$-discriminants in times similar or
faster to the techniques dedicated to this case described in
\cite{roseTropicalImplicitizationRevisited2025}.

Both of our algorithms are implemented in a software package
accompanying this paper, written in the programming language
\texttt{julia} \cite{bezanson2017}, and available at
\begin{center}
  \url{https://github.com/RafaelDavidMohr/MCISubdivisions.jl}.
\end{center}
We finally mention that all of our algorithmic techniques actually
work for so-called {\em matroid complete intersections} (MCI), a
notion introduced in \cite{esterovEngineeredCompleteIntersections2024,
  esterovEngineeredCompleteIntersections2025} to combinatorially
abstract ECI's in the same way as matroids combinatorially abstract
complements of hyperplane arrangements.

\paragraph*{Outline.}

\Cref{sec:prelim} introduces the necessary theoretical preliminaries
from \cite{esterovEngineeredCompleteIntersections2024,
  esterovEngineeredCompleteIntersections2025} in order to state our
algorithms. We therein also prove that the multiplicities associated
to the tropicalization of a positive dimensional MCI are given as
mixed volumes of zero-dimensional MCI's, this generalizes a result
mentioned above given in \cite{sturmfelsEliminationTheoryTropical}.

\Cref{sec:msd} introduces mixed subdivisions of MCI's and gives the
aforementioned tropical homotopy continuation algorithm to compute
them. We therein also recall the real patchworking results for ECI's
from \cite{esterovEngineeredCompleteIntersections2025} and show how
obtain algorithms from these results.

\Cref{sec:dte} concerns the elimination problem for ECI's described
above, again in the more general context of MCI's, and gives an
effective formula for vertices of the sought Newton polytope in terms
of the mixed subdivisions introduced in \Cref{sec:msd}.  This formula
can be combined with the technique introduced in
\cite{hugginsIB4eSoftwareFramework2006} in order to compute the full
Newton polytope of the sought eliminant.

Finally, \Cref{sec:examples} shows how our algorithms can be used in
the form of our software package on several classes of examples, in
particular for polynomial systems coming from chemical reaction
networks and for the computation of $A$-discriminants. As an
application of real patchworking for ECI's, and our algorithmic
version of this technique, we also prove in \Cref{sec:examples} that
there exists a degree 4 hypersurface in 3 variables such that all
singularities of its discriminant curve are real.

{\footnotesize \paragraph*{Acknowledgements.}
\label{par:ack}

Rafael Mohr was supported by the FWO grants G0F5921N (Odysseus) and
G023721N, and by the KU Leuven grant iBOF/23/064 and would like to
gratefully acknowledge this support. Yulia Mukhina was supported by
the French ANR-22-CE48-0008 OCCAM and ANR-22-CE48-0016 NODE projects
and would like to gratefully acknowledge this support. The authors
wish to thank Maximilian Wiesmann for helpful discussions.}

\section{Theoretical Preliminaries and Conventions}
\label{sec:prelim}

\subsection{Engineered and Matroid Complete Intersections}
\label{sec:ecis}

This section provides the tropical geometry background required for our work
and reviews the necessary definitions and results
from~\cite{esterovEngineeredCompleteIntersections2025,
  esterovEngineeredCompleteIntersections2024} in order to introduce
the central objects of this work --- mixed subdivisions of
Engineered Complete Intersections or ECI's (see ~\Cref{sec:msd}).
We assume some knowledge of tropical geometry on the part of the
reader and refer e.g. \cite{mikhalkinTropicalGeometryIts2006} for an
introduction to the subject and to
\cite{maclaganIntroductionTropicalGeometry2021} for a more
comprehensive treatment.

\begin{notation}
  Throughout this section, let $W$ be a finite-dimensional real
  vector space equipped with a distinguished lattice $L$ 
  (a discrete free subgroup in $W$ of rank $\dim(W)$).
\end{notation}

Later in the paper, $W$ and $L$ will be the dual spaces
$\RR^n$ and $\ZZ^n$, respectively.

\begin{definition}[Tropical Variety]
In the context of this paper, a {\em tropical variety} is a
pure-dimensional polyhedral fan in $W$ that is rational with respect
to $L$ and comes with a multiplicity attached to each maximal cone in
this fan that makes this fan balanced in the sense of Definition 3.3.1
in \cite{maclaganIntroductionTropicalGeometry2021}.
\end{definition}

\begin{notation}
  For a tropical variety $\tvar$ we use the notation $C\in \tvar$ 
  to refer to a top-dimensional cone in the fan underlying $\tvar$.
  We write $\mult(\tvar, C)$ for the multiplicity attached to $C$.

  For tropical varieties $\tvar_1$ and $\tvar_2$ we denote by
  $\tvar_1 \cdot \tvar_2$ the {\em stable intersection} of
  $\tvar_1, \tvar_2$ (for a definition of this construction we refer
  to \cite[Section~3.6]{maclaganIntroductionTropicalGeometry2021} or
  \cite[Section~4.3]{mikhalkinTropicalGeometryIts2006}).
\end{notation}

\begin{definition}[Refinement, Linear Span]
  \begin{enumerate}
    \item[]
    \item By a {\em refinement} of a tropical variety we mean a refinement of the
    underlying polyhedral fan with each maximal cone inheriting the
    multiplicity of the maximal cone of the original fan containing it. 
    \item By the {\em linear span} of
    a subset $S \subset W$ we mean the smallest linear subspace of $W$
    containing $S$ up to an affine shift.
  \end{enumerate}
\end{definition}

The tropicalizations that we will be interested in can be modeled by
{\em corner loci} of {\em conewise linear functions}:
\begin{definition}[Conewise Linear Function, Corner Locus]
  \label{def:cornerlocus}
  Let $\tvar$ be a tropical variety.
  \begin{enumerate}
  \item A {\em conewise linear function} $m: \tvar\rightarrow \RR$ is a
    continuous function such that there is a refinement $\tvar'$ of
    $\tvar$ such that $m$ is linear on each maximal cone of $\tvar'$
    and such that $m$ takes integral values on elements in
    $\tvar \cap L$.
  \item If $\tvar'$ is such a refinement then the {\em corner locus}
    of $m$ on $\tvar$, written $\corn(m\cdot \tvar)$, is constructed as
    follows: To each codimension one cone $C$ in $\tvar'$ we associate
    the multiplicty
  \[\mult(\corn(m\cdot \tvar), C) := \sum_{C\subset C'\in \tvar}m(v_{C'/C})\mult(\tvar,C')\]
  where $v_{C'/C}$ denotes a primitive generator of the image of $C'$
  modulo the linear span of $C$. The corner locus
  $\corn(m\cdot \tvar)$ is then the tropical variety given by those
  codimension one cones $C\in \tvar'$ with $\mult(\corn(m\cdot \tvar), C)\neq 0$.
  \end{enumerate}
\end{definition}
By a conewise linear function $m:W\rightarrow \RR$, we mean a conewise linear
function on $W$ considered as a tropical variety with the trivial
tropical structure (i.e. consisting only of $W$ itself with
multiplicity $1$). In this case we will write
$\corn(m) := \corn(m\cdot W)$.
\begin{remark}
  If $m :W\rightarrow \RR$ is a conewise linear function and
  $\tvar\subset W$ is a tropical variety then
  $\corn(m\cdot \tvar) = \corn(m) \cdot \tvar$ as shown in
  \cite{esterovTropicalVarietiesPolynomial2012}.
\end{remark}

\begin{definition}[Mixed Volume]
  \label{def:mv}
  Let $m_1,\hdots,m_k$ be conewise linear functions on $W$ and let
  $\tvar:=\corn(m_1)\cdot\hdots \cdot \corn(m_k)$. If $k\geq n$, the tropical variety
  $\tvar$ is either empty or consists just of the origin $0\in W$. The
  {\em mixed volume} of $m_1,\hdots,m_k$, written $\mv(m_1,\hdots,m_k)$,
  is defined as zero if $\tvar = \emptyset$ and as
  $\mult(\tvar, \set{0})$ otherwise.
\end{definition}
\begin{remark}
  Note that for a convex polytope $\Delta\subset \RR^n$, its {\em support
    function} $m_{\Delta}:(\RR^n)^{*}\rightarrow \RR$,
  $\omega\mapsto \max\setdes{\omega(p)}{p\in \Delta}$, is a conewise linear function on
  $(\RR^n)^{*}$. The nomenclature {\em mixed volume} in \Cref{def:mv}
  is justified by the fact that for $n$ convex polytopes
  $\Delta_1,\hdots,\Delta_n\subset \RR^n$ we have
  \[\mv(m_{\Delta_1},\hdots,m_{\Delta_n}) = \mv(\Delta_1,\hdots,\Delta_n)\]
  as was shown in \cite{esterovTropicalVarietiesPolynomial2012}.
\end{remark}

Let us recall how tropicalizations of ideals generated by Laurent
polynomials are constructed. We denote by $\CC[\xx^{\pm}]$ the
$n$-variate Laurent polynomial ring in the variables
$\xx:=\set{x_1,\hdots,x_n}$. For a Laurent polynomial
$f=\sum_{a\in A}c_a\xx^a\in \CC[\xx^{\pm}]$, where $A\subset \ZZ^n$ is finite, and
$\omega\in (\RR^n)^{*}$ we denote
$\deg_{\omega}(f) := \max\setdes{\omega(a)}{a\in A}$ and
\[\lt_{\omega}(f) := \sum_{\substack{a\in A\\ \omega(a) = \deg_{\omega}(f)}}c_aa.\]
For an ideal $I\subset \CC[\xx^{\pm}]$ we write $\lt_{\omega}(I):=\langle \lt_{\omega}(f)\;|\;f\in I\rangle$.
\begin{definition}[Tropicalization]
  The {\em tropicalization} of $I$, written $\trop(I)$, is constructed as follows:
  The underlying fan consists of the euclidean closures of the sets
  \[\setdes{\eta\in (\RR^n)^{*}}{\lt_{\omega}(I) = \lt_{\eta}(I)}\]
  where $\omega\in (\RR^n)^{*}$ is such that $\lt_{\omega}(I)$ does not contain a
  monomial. The multiplicities associated to these cones are constructed
  as in \cite[Section 3.4]{maclaganIntroductionTropicalGeometry2021}. 
\end{definition}

Finally we define engineered complete intersections, the class of polynomial systems
of our interest:
\begin{definition}[Engineered Complete Intersection]
  An {\em engineered complete intersection} (ECI) in $\CC[\xx^{\pm}]$
  consists of the following data:
  \begin{enumerate}
  \item A finite multiset $A\subset \ZZ^n$.
  \item A matrix $V\subset \CC^{k\times A}$, $k\leq n$, with columns indexed by
    $A$ and of maximal rank.
  \end{enumerate}
  We will write $v_i$ for the $i$th row of $V$. To the data $(V,A)$ we
  associate the polynomials
  \[f_i:=\sum_{a\in A}c_av_{i,a}\xx^a;\;i=1,\hdots,k,\] where the
  $c_a$, $a\in A$, are symbolic parameters. By a {\em specialization} of
  the ECI $(V,A)$ at $p\in \CC^A$ we mean the polynomials $f_i$ with the
  symbolic parameters $c_a$, $a\in A$, evaluated at $p$.
\end{definition}

\begin{remark}
  The name \textit{engineered complete intersection} is justified by
  the fact that, if one specializes an ECI at a generic $p\in \CC^A$,
  then one obtains a regular complete intersection \cite[Proposition
  4.3]{esterovEngineeredCompleteIntersections2024}. Note also that
  ECI's are often referred to as {\em vertically parametrized systems}
  in the literature.
\end{remark}

The theoretical basis of our algorithmic work was established in
\cite{esterovEngineeredCompleteIntersections2025,
  esterovEngineeredCompleteIntersections2024}: To an ECI $(V,A)$ in
$\CC[\xx^{\pm}]$ one can associate $k$ conewise linear functions such
that the tropicalization of a generic specialization of $(V,A)$
consists of the stable intersection of the corner loci of these $k$
conewise linear functions. Let us give the details of this
construction.

Given an ECI $(V,A)$, let $r:2^A\rightarrow \NN$ be the matroid on the columns
of $V$, mapping a subset $S\subset A$ to the rank of the submatrix of
$V$ indexed by the columns corresponding to $S$. The data $(r,A)$ was
called a matroid complete intersection in
\cite{esterovEngineeredCompleteIntersections2025,
  esterovEngineeredCompleteIntersections2024}:
\begin{definition}[Matroid Complete Intersection]
  A {\em matroid complete intersection} (MCI) $(r,A)$ consists of a
  finite multiset $A\subset L$ and a matroid $r:2^A\rightarrow \NN$. We refer to the
  quantity $r(A)$ as the {\em codimension} and to $n - r(A)$ as the
  {\em dimension} of $(r, A)$.
\end{definition}

To an MCI $(r,A)$ of codimension $k$ we can now associate $k$ conewise
linear functions on the vector space $W^{*}$ with distinguished
lattice $L^{*}$: Letting $m_1$ be the support function of the convex
hull of $A$, we construct $m_2,\hdots,m_k$ inductively. Having defined
$m_1,\hdots,m_{i-1}$ define
\[m_i(\omega) := \max\setdes{m}{r(\setdes{a\in A}{\omega(a)\geq m})\geq i}\] for
$\omega\in L^{*}$. Having defined $m_1,\hdots,m_k$ on $L^{*}$, we extend them
to $W^{*}$ by scalar extension.

Having constructed $m_1,\dots,m_k$, we now define tropicalizations
of MCI's:

\begin{definition}[Tropicalization of an MCI]
  \label{def:tropmci}
  \begin{enumerate}
  \item The {\em tropicalization of an MCI} $(r,A)$ with associated conewise
  linear functions $m_1,\dots,m_k$ is defined as
  \[\trop(r,A) := \corn(m_1)\cdot \hdots \cdot \corn(m_k).\]
  \item \label{defitem:mv1} If $k \geq n$, the dimension of $W$, then we
  define the {\em mixed volume} of $(r,A)$ as
  \[\mv(r,A):= \mv(m_1,\dots,m_k).\]
\item \label{defitem:mv2} More generally, if the dimension of the linear span $U$ of $A$
  is smaller or equal to $k$ then $\trop(r,A)$ is empty, in which
  case we define $\mv(r,A) = 0$, or consists of the orthogonal
  complement $U^{\bot}\subset W^{*}$ with a certain multiplicity which we also call
  the mixed volume $\mv(r,A)$ of $(r,A)$.
\item \label{defitem:tropmciseq} Yet more generally, consider several MCI's $(r_i,A_i)$,
  $i=1,\dots,\ell$, such that all $A_i$ can be shifted to the same vector
  subspace $U$ of $W$ with dimension less or equal to
  $\sum_ir_i(A_i)$. Then the stable intersection
  $\trop(r_1,A_1)\cdot \hdots \cdot \trop(r_{\ell},A_{\ell})$ is either empty, in which
  case we define $\mv(r_{\bullet},A_{\bullet}) = 0$, or consists of the orthogonal
  complement $U^{\bot}\subset W^{*}$ with a certain multiplicity which we call
  the mixed volume $\mv(r_{\bullet},A_{\bullet})$ of the sequence $(r_{\bullet},A_{\bullet})$.
  \end{enumerate}
\end{definition}

\begin{remark}
  \label{rem:mvremarks}
  Note that $\mv(r,A) = 0$ if $r(A) > \dim(W)$, including this case
  will be convenient later. Moreover, given a sequence of MCI's
  $(r_{\bullet},A_{\bullet})$ as in \Cref{defitem:tropmciseq} of \Cref{def:tropmci},
  the mixed volume $\mv(r_{\bullet},A_{\bullet})$ can still be computed as the
  mixed volume of a single MCI as follows: To see this, shift each
  $A_i\ni a_i$ to $\tilde A_i:=A_i-a_i\ni 0$ with rank functions
  $\tilde r_i(a):=r_i(a+a_i)$. Then
  $\mv(\sum_i\tilde r_i, \bigsqcup_i\tilde A_i)$ does not depend on the
  choice of shifts $a_i$, and equals the mixed volume of
  the sequence $r_{\bullet},A_{\bullet}$ as in \Cref{def:tropmci}.
\end{remark}

\begin{example}
  \label{ex:hexagon}
  Suppose $A$ consists of the vertices of the hexagon in $\ZZ^2$:
  
  \begin{center}
    \begin{tikzpicture}
      \draw[line width = 1pt][-, Black] (0,0) -- (1,0) -- (2,1) -- (2,2) -- (1,2) -- (0,1) -- cycle;

      \draw[line width = 1pt][-, Red] (0,0) -- (0,1);
      \draw[line width = 1pt][-, Blue] (1,2) -- (2,2);
      \draw[line width = 1pt][-, Green] (1,0) -- (2,1);

      \node[circle, draw=Green, fill=Green, inner sep=0pt, minimum size=0.4em] at (1,0) { };
      \node[below] at (1,0) {$a_2$};
      \node[circle, draw=Green, fill=Green, inner sep=0pt, minimum size=0.4em] at (2,1) { };
      \node[right] at (2,1) {$a_3$};
      \node[circle, draw=Blue, fill=Blue, inner sep=0pt, minimum size=0.4em] at (2,2) { };
      \node[above] at (2,2) {$a_4$};
      \node[circle, draw=Blue, fill=Blue, inner sep=0pt, minimum size=0.4em] at (1,2) { };
      \node[above] at (1,2) {$a_5$};
      \node[circle, draw=Red, fill=Red, inner sep=0pt, minimum size=0.4em] at (0,1) { };
      \node[left] at (0,1) {$a_6$};
      \node[circle, draw=Red, fill=Red, inner sep=0pt, minimum size=0.4em] at (0,0) { };
      \node[below] at (0,0) {$a_1$};
    \end{tikzpicture}
  \end{center}

  We take $V$ to consist of two generic linear combinations of three
  fixed generic binomials supported on the edges marked in red, green
  and blue. The first conewise linear function $m_1$ associated to the
  resulting MCI $(r,A)$ coincides with the support function of the
  convex hull of $A$.  Its corner locus $\corn(m_1)$ thus consists of
  the outer normal vectors of the edges of the hexagon, each with
  multiplicity $1$. The second support function $m_2$ evaluates to
  zero on those rays of $\corn(m_1)$ corresponding to the colored
  edges of the hexagon and to $1$ on primitive generators of the
  remaining rays. Thus, following \Cref{def:cornerlocus}, we find
  $\mv(r,A) = 3$. Note that the normalized volume of the convex hull
  of $A$ is $6$, this quantity coincides with the number of joint
  solutions in $(\CC^{*})^2$ of two generic polynomials with support
  $A$ by Koushnirenko's theorem \cite{kouchnirenko1976}.
\end{example}

The relevance of \Cref{def:tropmci} is the following theorem:
\begin{theorem}[Section 4 in \cite{esterovEngineeredCompleteIntersections2024}]
  \label{thm:tropvps}
  Let $I_p$ be the ideal generated by the specialization of an ECI
  $(V,A)$ at $p\in \CC^A$. Let $(r,A)$ be the MCI constructed from
  $(V,A)$ as above. If $p$ is sufficiently generic then
  \[\trop(I_p) = \trop(r,A).\]
\end{theorem}

In particular, if $(V,A)$ is an ECI with $V$ consisting of $n$ rows
then the root count of a generic specialization of $(V,A)$ matches the
mixed volume of the associated MCI $(r,A)$. Computing this quantity as
a sum of volumes of convex polytopes will be the subject of
\Cref{sec:msd}.

\subsection{Tropical Multiplicities of Matroid Complete Intersections}
\label{sec:weights}

Again, let $W$ be a real vector space of dimension $n$ together with a
distinguished lattice $L\subset W$. Given an MCI $(r,A)$ with
$A\subset L$ and of positive dimension, we aim now to describe the
multiplicities of the maximal cones in $\trop(r,A)$. It turns out that
these multiplicities are themselves mixed volumes of certain
zero-dimensional MCI's derived from $(r,A)$, defined as follows:

\begin{definition}[Cancellation]
  \label{def:canc}
  For $\omega\in L^{*}$, define the {\em $\omega$-cancellation}
  $\canc_{\omega}(r,A)$ of $(r,A)$ as the following recursively constructed
  sequence of MCI's $(r^{\omega}_i,A_i^{\omega})$, $i = 1,\dots,\ell$:
  \begin{enumerate}
  \item $A^{\omega}_1\subset A$ is the multiset of all $a\in A$ at which
    $\omega$ attains its maximum on $A$;
  \item $r^{\omega}_1$ is the restriction of $r$ from $2^A$ to $2^{A^{\omega}_1}$;
  \item Let $q$ be the matroid quotient of $r$ by $A_1^{\omega}$ and let
    $B:=\setdes{a\in A}{q(a)> 0}$.  We set $\ell = 1$ if
    $B= \emptyset$ and otherwise let $(r^{\omega}_i,A_i^{\omega})$,
    $i = 2,\dots,\ell$, be the sequence of $\omega$-cancellations of the MCI
    $(q,B)$.
  \end{enumerate}
\end{definition}

Recall that the \textit{matroid quotient} $q$ of $r$ by a subset $B\subset A$
is defined as $q : 2^A\rightarrow \NN$, $q(S) = r(B\cup S) - r(B)$.

The nomenclature {\em $\omega$-cancellation} in \Cref{def:canc} is justified
as follows: Suppose that $(r,A)$ arises as the MCI associated to an
ECI $(V,A)$ with $V\in \CC^{k\times A}$ and $A\subset \ZZ^n$. Let
$\omega\in (\RR^n)^{*}$ and let $\tilde{V}\in \CC^{k\times A}$ be the matrix
obtained by first taking as its rows $k$ generic linear combinations
of the rows of $V$ and then by reordering the columns of the resulting
matrix descendingly by $\omega$-degree. Let $\tilde{V}_{\text{red}}$ be a
row echelon form of $\tilde{V}$ computed without swapping rows or
columns. Each row of $\tilde{V}_{\text{red}}$ corresponds to a Laurent
polynomial with support in $A$ and $m_i(\omega)$ coincides with the value
of the support function of the Newton polytope of the $i$th row of
$\tilde{V}_{\text{red}}$ at $\omega$. The rows of
$\tilde{V}_{\text{red}}$ can be grouped according to $\ell$ different
$\omega$-degrees $d_1,\dots,d_{\ell}$ and the $j$th MCI
$(r^{\omega}_j,A_j^{\omega})$ is the MCI associated to the ECI given by the
$\omega$-leading terms of those rows of $\tilde{V}_{\text{red}}$ with
$\omega$-degree $d_j$.

\begin{example}
  \label{ex:matroidcurve}
  If $(r,A),\,A\subset\ZZ^n$, is a 1-dimensional MCI, then
  the mixed volume of the cancellations $(A^\omega_{\bullet},r^\omega_{\bullet})$ is defined for
  every $\omega\in(\ZZ^n)^*$, and equals non-zero numbers $m_j>0$ for at most
  finitely many primitive covectors $\omega_j$.
\end{example}

We can now describe the multiplicities of maximal cones in the
tropicalization of an MCI as mixed volumes of cancellations:

\begin{proposition}
  \label{prop:cancel}
  \begin{enumerate}
  \item \label{ione:cancel} For a generic $\omega$ in a cone
    $C\subset (W)^*$ (i.e. avoiding finitely many proper linear subspaces),
    each support set $A^{\omega}_i$ of the $\omega$-cancellation
    $\canc_{\omega}(r,A)$ is contained in the orthogonal complement
    $C^\bot\subset W$ up to an affine shift.
  \item \label{itwo:cancel} In a small neighborhood of every such
    $\omega$, $\trop(r,A)$ coincides with the stable intersection of the
    tropicalizations $\trop(r_i^{\omega},A_i^{\omega})$.
  \item If $\omega$ is a general covector in a full-dimensional cone
    $C\in \trop(r,A)$, then the mixed volume of
    $(A^\omega_{\bullet},r^\omega_{\bullet})$ is well defined by \cref{ione:cancel} and equals
    $\mult(C,\trop(r,A))$ by \cref{itwo:cancel}.
  \end{enumerate}
\end{proposition}
\begin{proof}
  For \cref{ione:cancel}, note that $(r,A)$ has only finitely many
  different cancellations $(q_{\bullet},B_{\bullet})$, and so $C$ is subdivided
  into finitely many polyhedral sets
  $C_{(q_{\bullet},B_{\bullet})}:=\setdes{\omega\in C}{\canc_{\omega}(r,A) =
    (q_{\bullet},B_{\bullet})}$. Hence a sufficient condition of general position for
  $\omega$ is to belong to one of those
  $C_{(q_{\bullet},B_{\bullet})}\subset C$ not contained in a hyperplane in
  $C$: then the respective $A^\omega_i$ shift to
  $C_{(q_{\bullet},B_{\bullet})}^\perp=C^\perp$ by definition. This is indeed a condition of general
  position, because the remaining
  $C_{(q_{\bullet},B_{\bullet})}\subset C$ belong to a finite union of hyperplanes in
  $C$.

  \Cref{itwo:cancel} follows from the definition of $\trop(r,A)$ and
  its cancellations $(r_{\bullet}^{\omega},A^\omega_{\bullet})$ as stable intersections
  $\corn(m_1)\cdot \hdots \cdot \corn(m_k)$ and
  $\corn(m_{1,i})\cdot \hdots \cdot \corn(m_{k_i,i})$ for suitable conewise
  linear functions $m_1,\hdots,m_k$ and $m_{1,i},\dots,m_{k_i,i}$
  respectively: By construction, the sequence of conewise linear
  functions $m_1,\hdots,m_k$, restricted to a small neighborhood of
  $\omega$, coincides with the concatenation of the sequences
  $m_{1,i},\ldots,m_{k_i,i}$ for $i = 1,\hdots, \ell$.
\end{proof}

\begin{example}
  \label{ex:matroidcurve1}
  \begin{enumerate}
  \item (\Cref{ex:matroidcurve} continued). The tropicalization of a
    1-dimensional MCI $(r,A),\,A\subset\ZZ^n$, consists of rays
    $\RR_+\cdot\omega_j$ with mixed volume multiplicities $m_j$, defined in
    \Cref{ex:matroidcurve}. The condition $m_j > 0$ enforces that for
    the corresponding cancellation
    $(r_{\bullet}^{\omega_j},A_{\bullet}^{\omega_j})$ we have
    $r_i^{\omega_j}(A_i^{\omega_j}) < | A_i^{\omega_j}|$ and that
    $\sum_ir_i^{\omega_j}(A_i^{\omega_j}) = n - 1$.
  \item If the tropicalization of an MCI $(r,A)$ is a plane, the
    multiplicity of this plane can naturally be called the {\em
      (generalized) mixed volume} of the MCI, even though $(r,A)$ may
    not satisfy the assumption of \cref{defitem:mv2} in
    \Cref{def:mv}. \Cref{prop:cancel} however represents it as the
    mixed volume (in the sense of \cref{defitem:tropmciseq} in
    \Cref{def:mv}) of the cancellations of $(r,A)$ in a general
    direction, so it can still be computed by the Algorithm that we
    will present in \Cref{sec:tropcont}.
  \end{enumerate}
\end{example}

\section{Mixed Subdivisions of Matroid Complete Intersections}
\label{sec:msd}

With the application of \Cref{thm:tropvps} in mind, we now introduce,
in \Cref{sec:mcims}, mixed subdivisions of MCI's by which their mixed
volumes can be expressed in terms of sums of volumes of convex
polytopes. We then give, in \Cref{sec:tropcont}, a tropical homotopy
continuation algorithm for computing such mixed subdivisions.

Throughout this section we again fix a real vector space $W$ of
dimension $n$ and a distinguished lattice $L$ in $W$.

We also fix a zero-dimensional MCI $(r,A)$ on $W$ such that the linear
span of $A$ is $W$.

\subsection{Mixed Cells and Mixed Subdivisions of MCI's}
\label{sec:mcims}

We now transport the notion of a {\em mixed cell} in \cite{huber1995}
to our context. 

\begin{definition}[Mixed Cell]
  \label{def:pmixedcell}
  A {\em partial mixed cell} of $(r,A)$ is a sequence of pairwise
  disjoint submultisets $M := (S_1,\hdots,S_k)$ of $A$ satisfying the
  following conditions:
  \begin{enumerate}
  \item The $S_i$ are mutually affinely independent, i.e. each $S_i$
    is affinely independent modulo the linear span of
    $S_1+ \hdots + S_{i-1} + S_{i+1}+\hdots + S_k$.
  \item $S_1$ is a circuit of the matroid $r$.
  \item $(S_2,\dots,S_k)$ is a partial mixed cell of the MCI $(q,B)$
    where $q$ is the matroid quotient of $r$ by $S_1$ and
    $B:=\setdes{a\in A}{q(a) > 0}$.
  \end{enumerate}
  We write
  \[\codim(M) := r(S_1\cup \hdots \cup S_k) = |S_1| + \hdots + |S_k| - k.\]
  If $\codim(M) = n$ then $M$ is simply called a mixed cell.
\end{definition}

\begin{remark}
  Note that the affine independence condition in \Cref{def:pmixedcell}
  enforces the $S_i$ to be subsets of $A$ and not submultisets,
  i.e. no $S_i$ contains multiple copies of the same element of the
  underlying set of $A$.
\end{remark}

\begin{example}
  Continuing with \Cref{ex:hexagon}, the tuple
  \[M := (\set{a_1,a_6}, \set{a_2,a_3})\] is a mixed cell of the
  associated MCI $(r, A)$. Note that e.g. $\set{a_1,a_2}$ does not give
  a partial mixed cell of $(r, A)$ as $\set{a_1,a_2}$ is not a circuit of
  $r$.
\end{example}

\begin{notation}
  For a function $d:A\rightarrow \RR$ and any subset $S\subset A$ write
  $S_d\subset W\times \RR$ for the graph of $d$ restricted to $S$. We denote by
  $(r_d,A_d)$ the MCI given by $A_d$ and with $r_d$ given by the
  composition of $r$ with the projection of $A_d$ to $A$.
\end{notation}

Occasionally, slightly abusing notation, we will consider a function
$d : A \rightarrow \RR$ as a vector $d\in \RR^A$ indexed by $A$ instead. We will
variously refer to $d$ as a \textit{height function} or \textit{height
  vector}.

\begin{notation}
  Throughout this paper, we denote for two integers $i,j\in \ZZ$
  \[\irng{i,j} = \setdes{\ell \in \ZZ}{i\leq \ell \leq j}.\]
\end{notation}

We now introduce the notion of a mixed subdivision of $(r, A)$. These
are associated to a height vector $d\in \RR^A$ in a similar way as regular
subdivisions of $A$, in a polyhedral sense.

\begin{definition}[Tropical Root, Dual Tropical Root, Mixed Subdivision]
  \label{def:tropicalroot}
  The mixed cell $M = (S_1,\hdots,S_k)$ is a {\em dual tropical root}
  of $(r,A)$ at the function $d:A\rightarrow \RR$ if only if there is a {\em
    tropical root} $\omega\in W^{*}$ such that
  $S_i=A_{d,i}^{(\omega,i)}$ for all $i\in \irng{1,k}$ where
  $\canc_{(\omega,1)}(r,A_d) =: (r_{d,i}^{(\omega,i)},A_{d,i}^{(\omega,i)})_i$.  The set
  of all dual tropical roots of $(r,A)$ at $d$ is called a {\em mixed
    subdivision} of $(r,A)$.
\end{definition}

\begin{notation}
  \label{not:volume}
  Fix a volume form $\mu$ on $W$ which evaluates to one on some basis of
  $L$. Let $M = (S_1,\hdots,S_k)$ be a mixed cell of $(r,A)$ with each
  $S_i$ of cardinality $s_i$. We denote
\[\vol(M):= (s_1-1)!\cdot\hdots \cdot (s_{k}-1)!\cdot \vol_{\mu}\left(\operatorname{conv}(S_1)+\hdots + \operatorname{conv}(S_k)\right)\]
where $\operatorname{conv}(S_i)$ is the convex hull of $S_i$.
\end{notation}

Note that, with this definition, $\vol(M)$ coincides with the mixed
volume
\[\mv(\underbrace{\operatorname{conv}(S_1),\hdots,\operatorname{conv}(S_1)}_{\text{$s_1-1$ times}},\hdots,\underbrace{\operatorname{conv}(S_k),\hdots,\operatorname{conv}(S_k)}_{\text{$s_k-1$ times}}),\]
see e.g. \cite{huber1995}.

We will now show that the
quantity $\mv(r,A)$ can be expressed in terms of volumes of dual
tropical roots at a generic height vector $d\in \RR^A$. We divide the
proof into two propositions.

Let us first describe $\trop(r_d,A_d)$. Note that $\trop(r_d,A_d)$
consists of finitely many rays, see \Cref{ex:matroidcurve}. If $\omega\in W^{*}$ is such that
$(\omega,1)$ is a generator of such a ray then, by \Cref{prop:cancel}, this
ray has multiplicity $\mv(\canc_{(\omega,1)}(r_d,A_d))$. Moreover,
\begin{proposition}
  \label{prop:liftroots}
  For any $d \in \RR^A$, the mixed volume of $(r,A)$ equals the sum of
  the mixed volumes $\mv(\canc_{(\omega,1)}(r_d,A_d))$ where
  $\omega\in W^{*}$ is such that $(\omega,1)$ generates a ray of $\trop(r,A_d)$.
\end{proposition}
\begin{proof}
  Choose a constant function $q<d$ on $A$ and define
  $\tilde A := A_d\cup A_q\subset L\times \ZZ$. Further, define
  $\tilde r$ on $2^{\tilde A}$ as the composition of the projection to
  $L$ with $r$.

  The tropicalization of the MCI
  $(\tilde r, \tilde A)$ consists of the following rays:
  \begin{itemize}
  \item Those rays in $\trop (r_d,A_d)$ which meet the upper halfspace
    $W\times\RR_{>0}$, with the same multiplicities as in $\trop(r_d,A_d)$.
  \item Certain rays (of no interest for us) in the boundary plane $H:=W\times\set{0}$.
  \item The ray $\set{0}\times\RR_{<0}$ with multiplicity $\mv(r,A)$, this is by choice
    of $q$.
  \end{itemize}

  By \cite[Proposition
  3.6.12]{maclaganIntroductionTropicalGeometry2021} the stable
  intersection of $\trop(\tilde r, \tilde A)$ with the boundary plane
  $H$, which consists of the origin with a certain multiplicity, can
  be computed in the following two ways: Either we shift $H$ to the
  lower halfspace, obtaining the multiplicity $\mv(r,A)$, or we shift
  $H$ to the upper halfspace in which case we obtain the sum of the
  mixed volumes in the statement of the proposition. Hence these two
  quantities must coincide.
\end{proof}

Next we have

\begin{proposition}
  \label{prop:mvsimplex}
  Suppose that $(r_i,S_i)_i$, $i=1,\dots,k$, is a sequence of MCI's in $W$ such that
  the $S_i$ are mutually affinely independent of cardinality $s_i$ and such that
  $r_i(S_i) = s_i - 1$. Then
  \begin{enumerate}
  \item $\mv((r_i,A_i)_i) = 0$ if one of the $S_i$ is not a circuit of $r_i$.
  \item
    $\mv(r,A) = (s_1-1)!\hdots (s_k-1)!\vol_{\mu}(S_1+\hdots + S_k)$
    otherwise.
  \end{enumerate}
\end{proposition}
\begin{proof}
  Any matroid $r$ on a groundset of $n+1$ elements and of rank $n$ is
  representable over $\CC$ as follows: Such a matroid has a unique
  circuit (say of rank $r$). Choose any linear subspace
  $U\subset \CC^{n}$ of dimension $r$ and $r + 1$ general vectors
  $u_1,\dots,u_{r+1}\in U$. Then choose $n-r$ general vectors
  $u_{r+2},\dots,u_{n+1}\in \CC^n$. The matroid $r$ is then isomorphic
  to the linear matroid on $u_1,\dots,u_{n+1}$.

  Hence, after choosing a suitable basis of $L$ and applying a linear
  isomorphism whose determinant has absolute value
  $(s_1-1)!\hdots (s_k-1)!\vol_{\mu}(S_1+\hdots + S_k)$, our sequence of
  MCI's is represented by an ECI consisting of a linear system of
  equations whose coefficient matrix has a block-diagonal structure,
  with each block corresponding to one of the $S_i$. This linear
  system has exactly one solution in the torus if each $S_i$ is a
  circuit of $r_i$ and no solution otherwise, proving the statement by
  \Cref{thm:tropvps}.
\end{proof}

Finally we are ready to show:

\begin{theorem}
  \label{thm:mv}
  If $d\in \RR^A$ is sufficiently generic, i.e. avoids finitely many
  linear subspaces of $\RR^A$, then we have
  \[\mv(r,A) = \sum_{\substack{M\text{ dual tropical root }\\\text{of }(r,A)\text{ at }d}}\vol(M).\]
\end{theorem}
\begin{proof}
  Suppose that $d$ is chosen so that for each
  ray of $\trop(r_d,A_d)$ with generator
  $(\omega,1)\in W^{*}\times \RR$, the support multisets
  $S_1,\dots,S_k$ of $\canc_{(\omega,1)}(r_d,A_d)$ are mutually affinely
  independent. This is ensured if $d$ lies outside finitely many linear
  subspaces of $\RR^A$. By \Cref{prop:cancel}, each such ray has multiplicity
  $\mv(\canc_{(\omega,1)}(r_d,A_d))$. By \Cref{ex:matroidcurve1},
  \Cref{prop:mvsimplex} and \Cref{rem:mvremarks} this mixed volume is
  zero unless $(S_1,\dots,S_k)$ is a mixed cell of $(r,A)$. Finally,
  from \Cref{prop:liftroots}, we obtain the desired equality.
\end{proof}

\subsection{Mixed Subdivisions of ECI's and Real Patchworking}
\label{sec:realpatch}

In \cite[section 5]{esterovEngineeredCompleteIntersections2025},
so-called {\em patchworking} techniques for ECI's were given,
extending the classical tropical patchworking techniques given in
\cite{viro1982gluing, sturmfelsVirosTheoremComplete1994}. These
patchworking techniques allow, for a given ECI $(V,A)$, to construct
ECI's derived from $(V,A)$ whose \textit{real} topology depends only
the set of dual tropical roots of the MCI $(r,A)$ corresponding to
$(V,A)$ at some height vector. Using the algorithmic techniques in
\Cref{sec:tropcont}, we will be able to use
these techniques to algorithmically construct ECI's with a prescribed
real root count.

We first recall the theoretical results of
\cite{esterovEngineeredCompleteIntersections2025} in a form suited to
our needs and then discuss how to obtain effective methods based on
them.

\subsubsection{Real Patchworking of ECI's}
\label{sec:realpatchone}

Let $\xx$ be a set of $n$ variables. Throughout this section, we
denote for a set of polynomials $g_1,\dots,g_r\in \RR[\xx^{\pm}]$
\[\Var_{\RR}(g_1,\dots,g_r) = \setdes{p \in (\RR^{*})^n}{g_1(p)=\ldots = g_r(p) = 0}.\]
Let $(V, A)$ be an ECI with $A\subset \ZZ^n$ and
$V\subset \RR^{n\times A}$, i.e. we assume that our ECI consists of
$n$ equations in $n$ variables. Let $d\in \RR^A$ be a generic height
vector (in the sense of \Cref{thm:mv}).
\begin{definition}[Engineering]
  The \textit{engineering} of $(V,A)$ at $d$ is the set of polynomials
  \[f_{i,d} := \sum_{a\in A}v_{i,a}\cdot t^{d_a} \cdot \xx^a\in \RR[\xx^{\pm}],\quad i = 1,\ldots, n.\]
\end{definition}
Let $(r,A)$ be the MCI associated to $(V,A)$ and let $\mathcal{M}$ be the set of
dual tropical roots of $(r,A)$ at height vector $d$. For
$M\in \mathcal{M}$, let $\omega\in (\RR^n)^{*}$ be the corresponding tropical root. As
explained after \Cref{def:canc}, the cancellation
$\canc_{\omega}(r,A)$ induces a set of polynomials
$f_1^M(\xx),\ldots,f_n^M(\xx)$, which are given as the
$\omega$-leading terms of the polynomials corresponding to the row echelon
form of $V$ with the columns of $V$ sorted descendingly by
$\omega$-degree. We now have the following combination of Theorem 5.4 and
Proposition 5.10 in \cite{esterovEngineeredCompleteIntersections2025}:
\begin{theorem}
  \label{thm:realrootcount}
  Let $f_{i,d}$, $i = 1,\ldots, n$, be the engineering of $(V,A)$ at
  sufficiently general $d\in \RR^A$. Then for $t_0\gg 0$
  \[\# \Var_{\RR}(f_{1,d}(t_0,\xx),\ldots,f_{n,d}(t_0,\xx)) = \sum_{M\in \mathcal{M}}\#
    \Var_{\RR}(f_1^M,\ldots, f_n^M). \]
\end{theorem}

\subsubsection{Counting Real Roots of Polynomial Systems with Simplicial Supports}
\label{sec:realpatchtwo}

A computational application of \Cref{thm:realrootcount} requires two
things: First computing the set $\mathcal{M}$ of dual tropical roots of an ECI
at a given height vector and second counting the number of real
solutions of the cancellations $f_1^M,\ldots, f_n^M$ for each $M\in \mathcal{M}$. The first
step will be treated in \Cref{sec:tropcont}. The second step can be performed
by using existing methods and software for computing real solutions of polynomial
systems, see e.g. \cite{berthomieu2021}. This is a challenging task in general.

It turns out, however, that a computationally much simpler procedure
exists for systems of the shape $f_1^M,\ldots, f_n^M$, making use of the
fact that the support of such systems consists of mutually transversal
simplices. We now describe this procedure.

To this end, we adopt the following notation for the rest of the section:
Let $S_i\subset\ZZ^n$, $i=1,\ldots,k$, be finite sets of cardinality
$$
|S_i|=d_i+1.
$$
Their disjoint union is understood as a union of labeled supports, so
the same exponent vector may occur in different sets $S_i$. We choose
an ordering on each $S_i$ by writing
$$
S_i=\{a_{i0},a_{i1},\ldots,a_{id_i}\}.
$$
We further assume that all the $S_i$ are mutually affinely independent, see
\Cref{def:pmixedcell}. This condition says that the sets $S_i$ are the
vertex sets of mutually transversal simplices.
Put
$$
D:=d_1+\cdots+d_k.
$$
In particular, $D\leq n$.

For every $i$, let
$$
\sigma_i:S_i\longrightarrow\{-1,+1\}
$$
be a sign function, and let
$$
U_i=(u_{ij}(a))_{\substack{1\leq j\leq d_i\\a\in S_i}}
$$
be a real $d_i\times(d_i+1)$ matrix of rank $d_i$. We study real roots in
$(\RR^{*})^n$ of the system of Laurent polynomials
\begin{equation}
  \label{eq:polsys}
  f_{ij}(x)=0, \qquad i=1,\ldots,k,\quad j=1,\ldots,d_i,
\end{equation}
where
\begin{equation*}
  f_{ij}(x)=\sum_{a\in S_i}\sigma_i(a)u_{ij}(a)x^a,\qquad x^a=x_1^{a_1}\cdots x_n^{a_n}.
\end{equation*}
We will characterize in terms of the support sets $S_i$ and the signs
of the maximal minors of $U_i$, whether the system \eqref{eq:polsys}
has solutions in a given orthant of $\RR^n$. For square systems of
equations of the shape \eqref{eq:polsys}, i.e. with $D=n$, this will
count the total number of real solutions in $(\RR^*)^n$, because each
orthant contains at most one solution.

We introduce the following further piece of notation: For every $i$,
define the \textit{cofactor vector} $c_i\in\RR^{S_i}$ via
$$
c_i(a_{ik})
:=
(-1)^k\det U_i[A_i\setminus\{a_{ik}\}],
\qquad
k=0,\ldots,d_i.
$$
Here, $U_i[A_i\setminus\{a_{ik}\}]$ denotes the matrix $U_i$ with the column corresponding
by $a_{ik}$ removed.
The cofactor identity gives
$$
U_ic_i=0.
$$
Since $U_i$ has rank $d_i$, its kernel is one-dimensional, and therefore
$$
\ker U_i=\RR c_i.
$$

We start our characterization of the number of real roots of
\eqref{eq:polsys} with the following simple case:
\begin{lemma}\label{ltriv}
  Once any maximal minor of any of the matrices $U_i$ is zero, the
  system of equations \eqref{eq:polsys} has no solutions in $(\RR^{*})^n$.
\end{lemma}
\begin{proof}
  If a maximal minor of $U_i$ vanishes, then one coordinate $c_i(a)$
  is zero. Every vector in $\ker U_i$ then has $a$-coordinate equal to
  zero. On the other hand, for every $x\in(\RR^*)^n$, the coefficient
  $\sigma_i(a)x^a$ is nonzero. Hence the coefficient vector
$$
(\sigma_i(a)x^a)_{a\in A_i}
$$
cannot belong to $\ker U_i$. Thus the equations belonging to the
$i$-th block of \eqref{eq:polsys} cannot be satisfied.
\end{proof}
We note that systems of the form \eqref{eq:polsys} arising as
leading terms of cancellations of ECI's never satisfy the condition of
\Cref{ltriv}. This is because of the third condition in
\Cref{def:pmixedcell}. We therefore now study the case where the
condition of \Cref{ltriv} is not satisfied.

To this end, let
$$
b_i(a_{ik})\in\FF_2 \quad \text{ so that }
\quad (-1)^{b_i(a_{ik})}
=
\sgn\left(\sigma_i(a_{ik}) c_i(a_{ik})\right),
\qquad
k=0,\ldots,d_i.
$$
Note that reordering $S_{i}$ either preserves all values of $b_i$ or
alters all of them. The subsequent usage of $b_i$ is indifferent to
this change. More precisely, let $\FF_2^{S_i}$ be the vector
space of functions $S_i\to\FF_2$, and let
${\rm const}_i\simeq\FF_2$ be the subspace of constant
functions on $S_i$. Then the class $[b_i]$ of the function
$b_i:S_i\to\FF_2$ in the quotient
$$
H_i:=\FF_2^{A_i}/{\rm const}_i
$$
does not depend on the ordering of $S_i$.

The set $E$ of orthants in $\RR^n$ can be identified with
the vector space
$\mathbb F_2^n$ by associating to
$e=(e_1,\ldots,e_n)\in \mathbb F_2^n$ the orthant
$$
\mathcal O_e
=
\setdes{x\in (\RR^{*})^n}{(-1)^{e_\ell}x_\ell>0\text{ for all }\ell}.
$$

We will define an $\FF_2$-linear map $\Lambda$ from $E$ to
$H:=\bigoplus_{i=1}^r H_i$, depending only on the sets $S_i$, whose
fiber over $$b:=([b_1],\ldots,[b_r])\in H$$ coincides, under this
identification, with the sought set $\{$orthants with a solution of
the system \eqref{eq:polsys}$\}$. Reducing the elements contained in
the $S_i$ modulo $2$, we define $\Lambda:E\longrightarrow H$ as the sum of
$$
\Lambda_i:E\longrightarrow H_i,
\qquad
\Lambda_i(e)
:=\left[
a\longmapsto\sum_{\ell=1}^n a_\ell e_\ell
\right].
$$

We finally have the following result:

\begin{theorem}\label{thm:rootcounting}
  Assume that the system of equations \eqref{eq:polsys} does not
  satisfy the conditions of \Cref{ltriv}. Under the identification of
  $\FF_2^n$ with the group of orthants $E$ given above, we have
  \[\setdes{e\in E}{\mathcal{O}_e\text{ contains a solution of \eqref{eq:polsys}}} \cong \Lambda^{-1}(b).\]

  Consequently, the total number of orthants with a solution of
  \eqref{eq:polsys} (and thus the total number of real toric solutions
  for square systems of the shape \eqref{eq:polsys}) not satisfying the
  condition of \Cref{ltriv}) equals
  \begin{itemize}
  \item $0$ if $b$ is not contained in the image of $\Lambda$;
  \item $2^{n-\operatorname{rank}_{\mathbb F_2}\Lambda}$ otherwise.
  \end{itemize}
\end{theorem}
\begin{proof}
  By assumption, every coordinate of every cofactor vector $c_i$ is
  nonzero. Fix $e\in E$ and define
  \begin{align*}
    q_\ell&:=(-1)^{e_\ell}\\
    q&:=(q_1,\ldots,q_n),
  \end{align*}
  and consider the corresponding orthant $\mathcal O_e$. Every
  $x\in\mathcal O_e$ can be written uniquely as $x = q\cdot y$ for some
  $y\in\RR_{>0}^n$ under coordinatewise multiplication. Then, for every
  $a\in \ZZ^n$, we have
  \[x^a = q^ay^a\quad \text{ where }\quad q^a=(-1)^{\sum_{\ell=1}^n a_\ell e_\ell}.\]

  For every block of \eqref{eq:polsys} indexed by $S_i$, the equations
  $f_{ij}(x)=0$, $j=1,\ldots,d_i$, are equivalent to
  $$
  \left(\sigma_i(a)x^a\right)_{a\in S_i}\in\ker U_i=\RR c_i.
  $$
  Hence if $x$ is a solution of \eqref{eq:polsys} then for every $i$
  there is a sign $\eta_i\in\{\pm1\}$ such that
  $$
  \sigma_i(a)q^a
  =
  \eta_i\sgn c_i(a)
  $$
  for every $a\in S_i$. Equivalently,
  $$
  q^a\sgn\bigl(\sigma_i(a)c_i(a)\bigr)
  $$
  is independent of $a\in A_i$.
  By the definition of $b_i$, this says that the function
  $$
  a\longmapsto \left[\sum_{\ell=1}^n a_\ell e_\ell+b_i(a)\right]
  $$
  is constant on $A_i$. Since addition and subtraction coincide in
  $\mathbb F_2$, this is equivalent to
  $$
  \left[
    a\longmapsto\sum_{\ell=1}^n a_\ell e_\ell
  \right]
  =
  [b_i].
  $$
  Requiring this for every $i$ gives
  $$
  \Lambda(e)=b.
  $$
  Thus every solution-containing orthant of $\RR^n$ yields an
  element of $\Lambda^{-1}(b)$.

  Conversely, suppose that $\Lambda(e)=b$. Then, for every $i$, the function
  $$
  a\longmapsto \left[\sum_{\ell=1}^n a_\ell e_\ell+b_i(a)\right]
  $$
  is constant on $S_i$. Consequently, there is a sign
  $\eta_i\in\{\pm1\}$ such that
  $$
  s_i(a)q^a
  =
  \eta_i\sgn c_i(a)
  $$
  for every $a\in S_i$. It remains to construct the absolute value of
  the root corresponding to $e$. Consider the linear map
  \begin{align*}
    &L:\RR^n\longrightarrow\RR^D,\\
    &L(v) := \left(
    \langle a_{ik}-a_{i0},v\rangle\right)_{\substack{1\leq i\leq r\\1\leq k\leq d_i}}.  
  \end{align*}
  The vectors $a_{ik}-a_{i0}$ are pairwise linearly independent by the
  transversality assumption on the $S_i$. Therefore $L$ has rank $D$ and
  is surjective. We can thus choose $v\in\RR^n$ satisfying
  $$
  \langle a_{ik}-a_{i0},v\rangle
  =
  \log\frac{|c_i(a_{ik})|}{|c_i(a_{i0})|}
  $$
  for every $i$ and every $k=1,\ldots,d_i$. Put
  $$
  y_\ell=e^{v_\ell},
  \qquad
  \ell=1,\ldots,n.
  $$
  Then, for every $i$, there is a positive number $\mu_i$ such that
  $$
  y^a=\mu_i|c_i(a)|
  $$
  for all $a\in S_i$. Finally setting $x:=q\cdot y\in\mathcal O_e$, we obtain
  $$
  \sigma_i(a)x^a
  =
  \sigma_i(a)q^ay^a
  =
  \eta_i\mu_i c_i(a).
  $$
  Hence
  $$
  \left(\sigma_i(a)x^a\right)_{a\in S_i}
  =
  \eta_i\mu_i c_i\in\ker U_i,
  $$
  and so $x$ is a solution of \eqref{eq:polsys} contained in $\mathcal O_e$. We have
  thus finally proven
  $$
  e\in\Lambda^{-1}(b)
  \quad\Longleftrightarrow\quad
  \mathcal O_e\text{ contains a solution of \eqref{eq:polsys}}.
  $$
  Because $e\mapsto\mathcal O_e$ is a bijection from $\FF_2^n$ to the group
  of orthants of $\RR^n$, it restricts to a bijection from the fiber
  $\Lambda^{-1}(b)$ to the set of orthants containing a solution of
  \eqref{eq:polsys}.

  If $b\notin\operatorname{im}\Lambda$, this fiber is empty. Otherwise it is an
  affine space in $\FF_2^n$ of dimension
  $\dim\ker \Lambda = n-\operatorname{rank} \Lambda$. Such an affine
  space has cardinality $2^{n-\operatorname{rank}_{\mathbb F_2}\Lambda}$. 

  Finally, if $D=n$, the map $L$ is an isomorphism. Thus the vector
  $v$, and hence the solution $x$ constructed above in a prescribed orthant, is
  unique. Therefore, for square systems, the number of real toric solutions equals
  the number of solution-containing orthants.
\end{proof}

\Cref{thm:rootcounting} allows us to easily evaluate the right hand
side of the equality in \Cref{thm:realrootcount} given a mixed
subdivision of an ECI by solving a linear system. This will be applied
to an example in \Cref{sec:examples}.

\subsection{Tropical Homotopy Continuation for MCI's}
\label{sec:tropcont}

Next, we want to design an algorithm that computes a mixed subdivision
of the MCI $(r,A)$ at some given sufficiently generic height vector
$d\in \RR^A$. For this we design a {\em tropical homotopy continuation}
algorithm similar to \cite{jensenTropicalHomotopyContinuation2016,
  malajovichComputingMixedVolume2017a,
  daiseyFrameworkGeneralizedTropical2024}. The core idea is the
following: Given a mixed subdivision of $(r,A)$ at some other height
vector $d'\in \RR^A$ we follow the continuous path
$d(t):=(1-t)d' + td$ from $t = 0$ to $t = 1$, detecting at which
$t\in [0,1]$ the given mixed subdivision at height $d(t)$ changes. As in
the other tropical homotopy continuation algorithms mentioned above,
this relies crucially on the fact that, given a mixed cell $M$ of
$(r,A)$, the set of all $d\in \RR^A$ such that $M$ is a dual tropical
root of $(r,A)$ at $d$ forms the relative interior of a convex
polyhedral cone defined by linear inequalities which can be explicitly
stated. We now give the description of this convex polyhedral cone which
we call a {\em mixed cell cone}.

We assume for the remainder of this section that $W = \RR^n$ and
$\ZZ^n =: L\subset W$ is the lattice spanned by the standard unit vectors.

The linear inequalities defining mixed cell cones arise from vectors
of the following kind:

\begin{definition}[Affine Relation, Affine Circuit]
  Let $A_1,\hdots,A_k\subset \ZZ^n$ be finite multisets. An {\em affine relation}
  of $(A_1,\hdots,A_k)$ is a vector $c\in \RR^{A_1}\times \hdots \times \RR^{A_k}$ such that
  \begin{align*}
    &\sum_{j=1}^k\sum_{a\in A_j}c_aa = 0\text{ and}\\
    &\sum_{a\in A_j}c_a=0\; \forall j\in \irng{1,k}.
  \end{align*}
  The vector $c$ is called an {\em affine circuit} if, for each
  $j\in \irng{1,k}$, $c$ induces a minimal affine dependence relation of
  $A_j$ modulo the linear spans of
  $A_1,\ldots,A_{j-1},A_{j+1},\ldots, A_k$.  If the $A_i$ are all contained in a
  finite multiset $A\subset \ZZ^n$ and $c$ is an affine relation of
  $(A_1,\hdots,A_k)$ then we write $\pi_A(c)$ for the image of $c$ under
  the canonical projection $\RR^{A_1}\times \hdots \times \RR^{A_k}\rightarrow \RR^A$.
\end{definition}

\begin{remark}
  Note that if $a\in A$ appears more than once than we consider the
  vector $c\in \RR^A$ with entry $1$ at one copy of $a$, entry $-1$ at
  another copy of $a$ and all other entries equal to $0$ as an affine
  circuit of $A$.
\end{remark}

Let $M = (S_1,\dots,S_k)$ be a mixed cell of $(r,A)$. We write
\[A_i:=\setdes{a\in A}{r(S_1\cup \hdots \cup S_{i-1}\cup \set{a}) = r(S_1\cup \hdots \cup S_{i-1}) + 1}.\]
 Due to the affine independence condition on $M$, for every
$a\in A_i\setminus S_i$, there is a unique affine circuit $c^{(i,a)}$ of
$(S_1,\hdots,S_i\cup \set{a},S_{i+1},\hdots,S_k)$ such that
$c^{(i,a)}_a = -1$. Moreover, there is a unique affine circuit
$c^{(i,i+1)}$ of $(S_1,\hdots,S_i\cup S_{i+1},\hdots,S_k)$ such that for
any $s\in S_{i+1}$ we have $\pi_A(c^{(i,i+1)}) = \pi_A(c^{(i,s)})$. Having
introduced this notation, we can now describe the set of all
$d\in \RR^A$ such that $M$ is a dual tropical root of $(r,A)$ at $d$, i.e.
the mixed cell cones:

\begin{lemma}
  \label{thm:mixedcone}
  For a mixed cell $M$ of $(r,A)$, the euclidean closure $C_M(r,A)$ of the set
  \[C_M^{\circ}(r,A)=\setdes{d\in \RR^A}{M\text{ is a dual tropical root of }(r,A)\text{ at }d}\]
  is a polyhedral cone, equal to the intersection of the two cones
  \[\setdes{d\in \RR^A}{\pi_A(c^{(i,a)})\cdot d \geq 0 \quad \forall a \in A_i\setminus (A_{i+1}\cup S_i)\forall i\in \irng{1,k}}\]
  and
  \[\setdes{d\in \RR^A}{\pi_A(c^{(i,i+1)})\cdot d \geq 0 \quad \forall i\in \irng{1,k-1}}.\]
\end{lemma}
\begin{proof}
  By the definition of dual tropical roots, the mixed cell $M$ is a
  dual tropical root of $(r,A)$ at $d$ if and only if, for the outer
  normal vector $(\omega,1)\in W^{*}\times \RR^{*}$ of the hyperplane going
  through $S_{1,d}+\hdots +S_{k,d}$, $S_{1,d},\hdots,S_{k,d}$
  coincides with the sequence of the underlying support sets of the
  cancellation $\canc_{(\omega,1)}(r,A)$. This is the case if and only if
  \begin{itemize}
  \item $(\omega,1)$ attains its maximum in $A_{i,d}$ exactly $S_{i,d}$ for every $i\in \irng{1,k}$ and
  \item we have $(\omega,1)(S_{1,d}) > (\omega,1)(S_{2,d}) > \hdots > (\omega,1)(S_{k,d})$.
  \end{itemize}
  The first condition is met if the second condition is met and
  $(\omega,1)$ attains its maximum in $A_{i,d}\setminus (A_{i+1,d}\cup S_{i+1,d})$
  exactly at $S_{i,d}$.

  Just as in the proof of Lemma 4.4 in
  \cite{jensenTropicalHomotopyContinuation2016}, these two conditions are hence met
  if and only if $d$ satisfies the linear inequalities
  \[\pi_A(c^{(i,a)})\cdot d > 0\quad \forall a \in A_i\setminus (A_{i+1}\cup S_i)\forall i\in \irng{1,k}\]
  and
  \[\pi_A(c^{(i,i+1)})\cdot d \geq 0 \quad \forall i\in \irng{1,k-1}.\]
\end{proof}

\begin{remark}
  Note that, in the notation of the proof of \Cref{thm:mixedcone}, if
  $A_i=A\setminus (S_1\cup \hdots \cup S_{i-1})$ then we only need to consider the inequalities
  \[\pi_A(c^{(i,i+1)})\cdot d \geq 0 \quad \forall i\in \irng{1,k-1}\]
  and
  \[\pi_A(c^{(k,a)})\cdot d \geq 0 \quad \forall a \in A_k\setminus S_k.\]
\end{remark}

\begin{example}
  \label{ex:hexagonineqs}
  We continue with \Cref{ex:hexagon}. A mixed cell for the associated
  MCI $(r,A)$ is given by
  \[M = (S_1,S_2) := (\set{a_1,a_6},\set{a_2,a_3}).\] Let us work out
  the defining inequalities of the mixed cell cone associated to $M$,
  starting with the inequalities coming from $S_1$. In the notation of
  the proof of \Cref{thm:mixedcone} we have $A_1 = A$.  Note that
  $A\setminus (A_2\cup S_2) = \emptyset$. Hence, following the notation of
  \Cref{thm:mixedcone}, we only have to construct the inequality
  $\pi_A(c^{(1,2)})\cdot d \geq 0$ to construct the defining inequalities of
  $C_M(r,A)$ coming from $S_1$.

  This can be done as follows: let $p$ be the
  projection along the linear span of $S_2$. The projection $p(A)$
  looks as follows:
  \begin{center}
    \begin{tikzpicture}
      \draw[line width = 1pt][-, Black] (0,0) -- (3,0) -- (6,0);

      \node[circle, draw=Black, fill=Black, inner sep=0pt, minimum size=0.4em] at (0,0) { };
      \node[above] at (0,0) {$p(a_5) = p(a_6)$};
      \node[circle, draw=Black, fill=Black, inner sep=0pt, minimum size=0.4em] at (3,0) { };
      \node[above] at (3,0) {$p(a_1) = p(a_4)$};
      \node[circle, draw=Black, fill=Black, inner sep=0pt, minimum size=0.4em] at (6,0) { };
      \node[above] at (6,0) {$p(a_2) = p(a_3)$};
    \end{tikzpicture}
  \end{center}
  On $p(A)$ we affine the affine circuits
  \begin{align*}
    & 2p(a_1) - p(a_6) - p(a_2) = 0\\
    & 2p(a_1) - p(a_6) - p(a_3) = 0\\
  \end{align*}
  Lifting either of these circuits of $p(A)$ back to $A$ we obtain the desired
  inequality $\pi_A(c^{(1,2)})\cdot d \geq 0$:
  \begin{equation}
    \label{eq:ineqs1}
    \begin{aligned}
      &2d_1-d_6+d_3-2d_2 \geq 0\\
    \end{aligned}
  \end{equation}
  Finally we list the inequalities coming from $S_2$. The set $A_2$ is
  given by $A_2 = A\setminus S_1$.  Let now $p$ be the projection along the
  linear span of $S_1$. The projection $p(A_2)$ looks as follows:
  \begin{center}
    \begin{tikzpicture}
      \draw[line width = 1pt][-, Black] (0,0) -- (3,0);

      \node[circle, draw=Black, fill=Black, inner sep=0pt, minimum size=0.4em] at (0,0) { };
      \node[above] at (0,0) {$p(a_2) = p(a_5)$};
      \node[circle, draw=Black, fill=Black, inner sep=0pt, minimum size=0.4em] at (3,0) { };
      \node[above] at (3,0) {$p(a_3) = p(a_4)$};
    \end{tikzpicture}
  \end{center}
  For any element $a\in A\setminus S_1\cup S_2 = \set{a_4,a_5}$ we now get another
  defining inequality of $C_M(r,A)$. On $p(A)$ we have the affine
  relations
  \begin{align*}
    &p(a_2)-p(a_5) = 0\\
    &p(a_3)-p(a_4) = 0
  \end{align*}
  Lifting these back to $A$ we obtain the inequalities
  $\pi_A(c^{(2,a_5)})\cdot d \geq 0$ and $\pi_A(c^{(2,a_4)})\cdot d \geq 0$:
  \begin{equation}
    \label{eq:ineqs2}
    \begin{aligned}
      &d_2-d_5+2d_6-2d_1\geq 0\\
      &d_3-d_4+d_6-d_1\geq 0.
    \end{aligned}
  \end{equation}
  The inequalities given in \eqref{eq:ineqs1} and \eqref{eq:ineqs2}
  define the mixed cell cone $C_M(r,A)$.

  As a sanity check, we use the computer algebra system \texttt{Oscar}
  \cite{OSCAR-book} to generate a random vector in this cone. One such
  vector is given by
  \[d = (14, 27, 56, 63, 50, 27)\in \RR^A.\] The outer normal vector of
  the hyperplane going through $S_{1,d} + S_{2,d}$ is then given by
  $\omega = (-16, -13, 1)$. We then have
  \begin{align*}
    (\omega(a_{1,d}),\hdots, \omega(a_{6,d})) = (14, 11, 11, 5, 8, 14).
  \end{align*}
  This shows that indeed $M$ is a tropical root of $(r, A)$ at $d$.
\end{example}

For a mixed cell $M := (S_1,\hdots,S_k)$ denote
$\supp(M) = \bigcup_{i=1}^kS_i$. To give a tropical homotopy continuation
algorithm we next have to describe how the given mixed subdivision
$\mathcal{M}$ at height $d(t) = (1-t)d' + td$ changes when the path
$d(t)$ crosses a facet of a cone $C_M(r,A)$ where $M\in \mathcal{M}$. This is done by
the following corollary and its proof:

\begin{corollary}
  \label{cor:exchangerelation}
  Let $H_c:=\setdes{d\in \RR^A}{c(d) = 0}$ be one of the hyperplanes
  bounding a mixed cell cone as given in
  \Cref{thm:mixedcone}. For $d\in H_c$ and 
  $\varepsilon > 0$ let $\mathcal{M}_-$ be the set of dual tropical roots of
  $(r,A)$ at $d-\varepsilon c$. Similarly let $\mathcal{M}_+$ be the set of dual tropical
  roots of $(r,A)$ at $d + \varepsilon c$. Suppose that
  \begin{itemize}
  \item The number $\varepsilon$ is sufficiently small, in the sense that
    the continuos path between $d - \varepsilon c$ and $d + \varepsilon c$ crosses
    no facet of any mixed cell cone of any mixed cell in $\mathcal{M}_- \cup \mathcal{M}_+$
    except the one defined by $H_c$.
  \item The vector $d$ is sufficiently generic in the sense that it
    lies in the relative interior of the facet defined by $H_c$ for
    any mixed cell cone of any mixed cell in $\mathcal{M}_- \cup \mathcal{M}_+$.
  \end{itemize}

  Then, for every $M\in \mathcal{M}_+$ there is $M'\in \mathcal{M}_-$ such that
  \[\supp(M) \subseteq \supp(M') \cup \supp(c).\]
\end{corollary}
\begin{proof}
  The statement of this corollary is a tautology if
  $M\in \mathcal{M}_-$ as well.  Suppose now that this is not the case and let
  $M := (S_1,\hdots,S_k)$. Since we have $M\notin \mathcal{M}_+$, but
  $M\in \mathcal{M}_-$, there is a minimal $i\in \irng{1,k}$ such that, potentially
  after multiplying $c$ by a positive scalar,
  $\sum_{s\in S_i}c_s = 1$. Now one of the following two situations can
  arise: Using the notation of the proof of \Cref{thm:mixedcone},
  either $c = \pi_A(c^{(i,i+1)})$ or there is
  $a\notin A_i\setminus A_{i+1}$ such that $c = \pi_A(c^{(i,a)})$. Let
  $C := S_i\cup S_{i+1}$ in the first situation and
  $C := S_i\cup \set{a}$ in the second situation. We define
  $i_{\text{next}} = i+1$ in the first situation and
  $i_{\text{next}} = i + 2$ in the second situation.

  We now construct the desired mixed cell $M'\in \mathcal{M}_-$: We choose the
  first $i-1$ components of $M'$ to be $S_1,\hdots,S_{i-1}$.  Now
  choose $S_j'\subset C$ as any affinely independent circuit of the matroid
  quotient of $r$ by $S_1\cup \hdots \cup S_{i-1}$ with
  $\sum_{s\in S_i'} c_s < 0$. Such a circuit $S_i'$ exists in either
  situation and the complement $C\setminus S_i'$ is then an affinely
  independent circuit of the matroid quotient of $r$ by
  $S_1\cup \hdots \cup S_{i-1}\cup S_i'$. Now either $C\setminus S_i'$ is of
  cardinality at least two in which case we set
  \[M':= (S_1,\hdots,S_{i-1},S_i',C\setminus S_i',S_{i_{\text{next}}},\hdots,S_k)\]
  or this is not the case and then we define
  \[M':= (S_1,\hdots,S_{i-1},S_i',S_{i_{\text{next}}},\hdots,S_k).\]
  By the conditions on $\varepsilon$ and $d$ stated in the corollary,
  $M'$ is a mixed cell of $(r,A)$ with $M'\in \mathcal{M}_+$.
\end{proof}

The results of this section so far are summarized in \Cref{alg:thc}, giving
a tropical homotopy continuation algorithm for MCI's. The algorithmic core
of \Cref{alg:thc} is the procedure \textit{WallWalk}, which essentially follows
the proof of \Cref{cor:exchangerelation} and which is responsible for modifying
the current mixed subdivision when the considered homotopy path crosses a facet
of a mixed cell cone.

\begin{example}
  We continue with \Cref{ex:hexagonineqs}. For the mixed cell
  \[M = (S_1,S_2) = (\set{a_1,a_6}, \set{a_2,a_3}),\]
  one the inequalities of the associated mixed cell cone $C_M(r,A)$
  is given by
  \[2d_1 - d_6 + d_3 - 2d_2 \geq 0.\] In the notation of
  \Cref{thm:mixedcone}, this is the inequality
  $\pi_A(c^{(1,2)}) \cdot d \geq 0$. When the corresponding hyperplane
  $\pi_A(c^{(1,2)}) \cdot d = 0$ is crossed, in the sense of
  \Cref{cor:exchangerelation}, $M$ is replaced as follows, using the
  notation of the proof of \Cref{cor:exchangerelation}: The set
  $C$ is given by $C = \set{a_1,a_2,a_3,a_6}$. We have to identify
  the circuits of $r$ contained in $C$ such that the corresponding
  sum of coefficients of $\pi_A(c^{(1,2)})$ is negative. The only
  such circuit is given by $\set{a_2,a_3}$. Consequently, the mixed
  cell $M$ is replaced by
  \[M' := (\set{a_2,a_3}, \set{a_1,a_6}).\]
\end{example}

\begin{algorithm}
  \caption{Tropical Homotopy Continuation for MCI's}
  \label{alg:thc}
  \raggedright

  \begin{description}
  \item[Input] An MCI $(r,A)$, a vector $c\in \RR^A$ bounding a mixed cell cone, the set $\mathcal{M}_{-}$ of dual tropical roots of
    $(r,A)$ at $d-\varepsilon c$ for generic $d$ with $d\cdot c = 0$ and $\varepsilon>0$ sufficiently small, in
    the sense of \Cref{cor:exchangerelation}.
  \item[Output] The set of dual tropical roots $\mathcal{M}_+$ of $(r,A)$ at $d + \varepsilon c$.
  \end{description}

  \begin{pseudo}
    \kw{function} \fn{WallWalk}((r,A),c,\mathcal{M}_-)\\+
    $\mathcal{M}_+\gets \emptyset$\\
    for $M:=(S_1,\ldots,S_k)\in \mathcal{M}_-$\\+
    if $c$ \tn{defines a facet of} $C_M(r,A)$\\+
    if $c = \pi_A(c^{(i,a)})$ \tn{for some} $a\in A$\\+
    $C \gets S_i\cup \set{a}$\\
    $i_{\text{next}}\gets i + 1$\\-
    else\\+
    $C\gets S_i\cup S_{i+1}$\\
    $i_{\text{next}}\gets i + 2$\\-
    for \tn{any circuit} $S_i'$ \tn{of the matroid quotient of $r$ by $S_1\cup \ldots \cup S_{i-1}$} \tn{with} $\sum_{s\in S_i'}c_s< 0$\\+
    if $|C\setminus S_i'| \geq 2$\\+
    $\mathcal{M}_+\gets \mathcal{M}_+\cup \set{(S_1,\hdots,S_{i-1},S_i',C\setminus S_i',S_{i_{\text{next}}},\hdots,S_k)}$\\-
    else\\+
    $\mathcal{M}_+\gets  \mathcal{M}_+\cup \set{(S_1,\hdots,S_{i-1},S_i',S_{i_{\text{next}}},\hdots,S_k)}$\\---
    else\\+
    $\mathcal{M}_+\gets \mathcal{M}_+\cup \set{M}$\\--
    return $\mathcal{M}_+$
  \end{pseudo}

  \begin{description}
  \item[Input] An MCI $(r,A)$, a sufficiently generic $d'\in \RR^A$,
    the set $\mathcal{M}'$ of dual tropical roots of $(r,A)$ at
    $d'$, a sufficiently generic vector $d\in \RR^A$.
  \item[Output] The set of dual tropical roots $\mathcal{M}$ of $(r,A)$ at $d$.
  \end{description}
  \begin{pseudo}
    \kw{function} \fn{Homotopy}((r,A),d',\mathcal{M}',d)\\+
    $\mathcal{M}\gets \mathcal{M}'$\\
    $t_{\text{curr}}\gets 0$\\
    while \tn{there is} $t\in [t_{\text{curr}},1]$ \tn{with} $(1-t)d'+td$ \tn{on facet of a cone} $C_M(r,A)$, $M\in \mathcal{M}'$\\+
    $t_{\text{curr}}\gets $ \tn{the minimal such} $t$\\
    $c\gets$ \tn{defining equation of the crossed facet}\\
    $\mathcal{M}\gets \fn{WallWalk}((r,A),c,\mathcal{M})$\\-
    return $\mathcal{M}$
  \end{pseudo}
\end{algorithm}

\begin{remark}
  \label{rem:sympert}
  The correctness of the function \textit{Homotopy} defined in
  \Cref{alg:thc} follows from \Cref{cor:exchangerelation} and its
  proof as long as we can guarantee that $d'$ and $d$ are generic enough
  in the sense of \Cref{thm:mv} and so that the path $(1-t)d' + td$,
  $t\in [0,1]$, crosses only the relative interior of facets of mixed
  cell cones. This is, probabilistically, ensured if we take a general
  vector $v\in \RR^A$ and consider the path
  $(1-t)(d'+\varepsilon v) + td$ instead where $\varepsilon$ is sufficiently
  small. Instead of choosing a concrete $\varepsilon$ we work in practice with
  vectors $d'$ over the {\em dual numbers}
  $\DN := \RR[\varepsilon]/\langle \varepsilon^2\rangle$ instead which are totally ordered by
  $a + b\varepsilon < c + d\varepsilon$ iff $a < c$ or $a = b$ and $b < d$.  Another
  method for ensuring the required of $d$ and $d'$ is
  outlined in Section 6.1 of
  \cite{jensenTropicalHomotopyContinuation2016}.
\end{remark}

A natural question regarding \Cref{alg:thc} is the following: How can
we compute a mixed subdivision of an MCI $(r,A)$ at a height vector
$d\in \RR^A$ if no mixed subdivision at another height vector $d'$ is
available? To solve this problem we now transport and extend the
method described in Section 7.1 of
\cite{jensenTropicalHomotopyContinuation2016} to our setting. It
reflects a typical usage of homotopy methods in numerical polynomial
system solving, in the sense that we try to identify another MCI $(r',A')$
so that a mixed subdivision of $(r',A')$ is ``easy to compute'' and can
be deformed to obtain a mixed subdivision of $(r, A)$ at height $d$.

\begin{notation}
  \label{not:startsystem}
  Let $(r',B')$ be another MCI with $B\subset \ZZ^n$ and $r'(B) = n$. Suppose
  that the convex hull of $A$ is contained in the convex hull of $B$.
  Define the matroid $r''$ on $A\sqcup B$ as $r''(S) := \max\set{r(S\cap A) + r'(S\cap B),n}$.
  For a subset $S\subset A\sqcup B$ let $\mathbf{1}_S\in \RR^{A\sqcup B}$ be the vector which
  has entry $1$ at all $s\in S$ and entry $0$ everywhere else.
\end{notation}

With this notation we have
\begin{lemma}
  \label{lem:startsystem}
  Let $\mathcal{M}'$ be the dual tropical roots of $(r',B)$ at some height
  $v\in \RR^{B}$. Then for $\varepsilon > 0$ sufficiently small we have:
  \begin{enumerate}
  \item \label{stone} The set of dual tropical roots of $(r'',A\sqcup B)$ at height
    $\mathbf{1}_B + \epsilon \cdot (0,v)\in \RR^{A\sqcup B}$ coincides with $\mathcal{M}'$.
  \item \label{sttwo} Choose generic $w\in \RR^{A\sqcup B}$, in the sense of
    \Cref{thm:mv}, and let $w_A$ be the projection of $w$ to
    $\RR^A$. Let $\mathcal{M}''$ be the set of dual tropical roots of
    $(r'',A\sqcup B)$ at height
    $\mathbf{1}_A + \epsilon \cdot (w,0)\in \RR^{A\sqcup B}$. Then
    \[\mathcal{M}:=\setdes{(S_1,\dots,S_k)\in \mathcal{M}''}{S_j\subset A\;\forall j\in \irng{1,k}}\]
    is the set of dual tropical roots of $(r,A)$ at height $w_A$.
  \end{enumerate}
\end{lemma}
\begin{proof}
  If $M'$ is a dual tropical root of
  $(r', B)$ at height $v$, with corresponding tropical root
  $\omega'\in (\RR^n)^{*}$, then $M'$ is the sequence of
  support sets of the cancellation $\canc_{(\omega',1)}(r'_v,B_v)$. This
  remains true for the cancellation
  $\canc_{(\omega',1)}(r''_{d_0},(A\sqcup B)_{d_0})$ where
  $d_0 := \mathbf{1}_B + \epsilon \cdot v$. As $r''(S) = r'(S)$ for any subset
  $S\subset B$, $M'$ is a dual tropical root of $(r'',A\sqcup B)$ at height
  $d_0$. The condition $\operatorname{conv}(A)\subseteq \operatorname{conv}(B)$ forces that any dual
  tropical root $M''$ of $(r'',A\sqcup B)$ at height $d_0$ consists of a sequence
  of subsets of $B$. Again by the definition of $r''$, $M''$ then also is
  a dual tropical root of $(r',B)$. This proves item \ref{stone}.

  Item \ref{sttwo} is proven similarly.
\end{proof}

This lemma can be applied in particular by choosing $(r',B)$ as
$(u_{A,n},A)$ where $u_{A,n}$ is the uniform matroid of rank $n$ on
$A$, i.e. for every subset $S\subset A$ we have
$u_{A,n}(S) = \min\set{|S|,n}$. In this case the set of dual tropical
roots of $(u_{A,n},A)$ at some generic height $v\in \RR^A$ is simply
given by the maximal cells of the regular triangulation of $A$ at
height $v$ (see e.g. \cite{gelfand1994} for a definition), as the
circuits of $u_{A,n}$ are exactly the subsets of $A$ of cardinality
$n + 1$. We summarize this observation in \Cref{alg:msd}.

\begin{algorithm}
  \caption{Computing Mixed Subdivisions of MCI's}
  \label{alg:msd}
  \raggedright

  \begin{description}
  \item[Input] An MCI $(r,A)$, a height vector $d\in \RR^A$.
  \item[Output] The set of dual tropical roots of $(r,A)$ at height $d$.
  \end{description}

  \begin{pseudo}
    \kw{function} \fn{msd}((r,A),d)\\+
    $d' \gets $ \tn{any generic height vector in} $\RR^A$\\
    $\mathcal{M}' \gets $ \tn{the regular subdivision of} $A$ \tn{at height} $d'$\\
    $B \gets A$\\
    $r' \gets$ \tn{the uniform matroid of rank $n$ on $B$}\\
    $r''\gets$ \tn{the matroid defined in \Cref{not:startsystem}}\\
    $d_0 \gets \mathbf{1}_B + \varepsilon (0, d') \in \RR^{A\sqcup B}$\\
    $d_1 \gets \mathbf{1}_A + \varepsilon (d, 0) \in \RR^{A\sqcup B}$\\
    $\mathcal{M}'' \gets \fn{Homotopy}((r'',A\sqcup B), d_0, \mathcal{M}', d_1)$\\
    return $\setdes{(S_1,\dots,S_k)\in \mathcal{M}''}{S_j\subset A\; \forall j\in \irng{1,k}}$
  \end{pseudo}
\end{algorithm}

\section{Effective Tropical Elimination for ECI's}
\label{sec:dte}

Let $(V,A)$ be an ECI with $A\subset \ZZ^n\times \ZZ^{k}$ for integers
$k$ and $n$ and let $V\in \CC^{(n+1) \times A}$ be of rank $n + 1$. Let
$I_p$ be a specialization of $(V,A)$ and consider the
corresponding algebraic set
$\Var(I_p)\subset (\CC^{*})^n\times (\CC^{*})^{k}$. Then, if $p$ is
sufficiently generic, the closure of the projection of $\Var(I_p)$ to
$(\CC^{*})^{k}$ is a hypersurface $\Var(g)$ for a certain
$k$-variate Laurent polynomial $g$.

The tropicalization $\trop(g)$ of the ideal $\langle g \rangle$ coincides with the
corner locus of the support function of the Newton polytope of $g$,
which lies in $\RR^{k}$. Elimination commutes with tropicalization,
i.e. $\trop(g)$ is the projection of
$\trop(I_p)\subset \RR^n\times \RR^{k}$ to $\RR^{k}$ \cite[see
e.g.][]{sturmfels2008}. In particular, by \Cref{thm:tropvps}, the
Newton polytope of $g$ only depends on $(V, A)$ if $p$ is sufficiently
generic.

The goal of this section is to give an algorithm which computes this
Newton polytope, given $(V, A)$. We first give an effective formula
for the support function of the Newton polytope of $g$. This support
function formula is then used to extract a {\em vertex oracle} of the
Newton polytope of $g$, i.e. a function that, given some
$\omega\in (\RR^k)^{*}$, returns a vertex of the Newton polytope of
$g$ at which $\omega$ is maximized. This can then be combined with the
algorithm given in \cite{hugginsIB4eSoftwareFramework2006} to compute
the entire Newton polytope of $g$.

\subsection{Mixed Volumes of MCI's upon Symbolic Deformations}
\label{sec:deform}

We start with investigating how the mixed volume of an MCI changes
when this MCI is deformed linearly. This will be necessary for
obtaining the desired vertex oracle from the support function of the
sought Newton polytope.

\begin{definition}[Moving set]
  A \textit{$k$-variate moving set} in $\RR^{n+1}$ is a set of the form
  \[A(\ww) := \setdes{(a, \ell_a)}{a\in A}\] where
  $A\subset \ZZ^{n+1}$ is a finite multiset, and each
  $\ell_a := \ell_a(\mathbf{w})$ is a degree one form in symbolic variables
  $\mathbf{w} := \set{w_1,\dots,w_k}$.
\end{definition}

\begin{notation}
  For a moving set $A(\mathbf{w})$ we denote
  $A(\omega) := \setdes{(a, \ell_a(\omega))}{a\in A}$ for $\omega\in (\RR^k)^{*}$.
\end{notation}

\begin{proposition}
  \label{prop:movset}
  \begin{enumerate}
  \item \label{it:movsimplex} Let $A(\mathbf{w})$ be a $k$-variate moving simplex in
    $\RR^{n+1}$, i.e. suppose that $A(\mathbf{w})$ consists of $n + 2$
    elements. Then the function $\omega \mapsto \vol(\Conv(A(\omega)))$ is linear in a neighborhood
    of some $\omega_0\in (\RR^k)^{*}$ if the affine span of $A(\omega_0)$ has dimension $n + 1$.
  \item \label{it:movset} Let $A(\mathbf{w})$ be a $k$-variate moving
    set and choose $\omega_0\in (\RR^k)^{*}$.  Suppose that there is a
    collection of subsets $\setdes{S_i(\mathbf{w})}{i\in \irng{1,r}}$ of
    $A(\mathbf{w})$ defining a triangulation of $A(\omega)$ for all
    $\omega$ in a neighborhood of $\omega_0$, in the sense that each
    $S_i(\mathbf{w})$ is a moving simplex. Then
    $\omega \rightarrow \vol(\Conv(A(\omega)))$ is linear in a neighborhood of
    $\omega_0$.  This statement holds for every $\omega_0$ inside a dense open
    subset of $(\RR^k)^{*}$.
  \item \label{it:movsetmv} Let $A_i(\mathbf{w})$,
    $i \in \irng{1,n+1}$, be a collection of $k$-variate moving sets. If
    the
    function~$\omega \mapsto \vol(\Conv(\sum_{i\in I}A_{i}(\omega)))$ is linear for every
    $\omega$ in a neighborhood of some $\omega_0\in (\RR^k)^{*}$ for every
    $I\subseteq \irng{1,n+1}$ (see item \ref{it:movset}), then so is the
    function
    $\omega \mapsto \mv(A_1(\mathbf{w}),\ldots, A_{n+1}(\mathbf{w}))$. This statement
    holds for every $\omega_0$ inside a dense open subset of $(\RR^k)^{*}$.
  \end{enumerate}
\end{proposition}
\begin{proof}
  (\textit{Proof of item \ref{it:movsimplex}}) Choose
  $a(\mathbf{w})\in A(\mathbf{w})$ and let $M(\mathbf{w})$ be the matrix
  with columns $a'(\mathbf{w}) - a(\mathbf{w})$,
  $a'\in A(\mathbf{w})\setminus \set{a}$. Then, for any
  $\omega \in (\RR^k)^{*}$, the quantity $\vol(\Conv(A(\omega)))$ coincides with
  the absolute value of the determinant of $M(\omega)$.  This determinant
  depends linearly on $\omega$, as only the last row of $M$ depends on the
  variables $\mathbf{w}$. If $\omega_0$ is as in item \ref{it:movsimplex},
  then in a small neighborhood of $\omega_0$, the sign of this determinant
  is constant, proving the statement.

  (\textit{Proof of item \ref{it:movset}}) This statement follows from item
  \ref{it:movsimplex} as
  \[\vol(\Conv(A(\omega))) = \sum_{i=1}^r \vol(\Conv(S_i(\omega)))\]
  for every $\omega$ in a neighborhood of $\omega_0$. If
  $\omega_0\in (\RR^k)^*$ is sufficiently generic, i.e. lying inside a dense
  open subset of $(\RR^k)^{*}$, then we can find a collection of
  moving simplices $S_i(\mathbf{w})$, $i\in \irng{1,r}$, such that
  $S_i(\omega)$, $i\in \irng{1,r}$, defines a triangulation of
  $A(\omega)$ for every $\omega$ in a small neighborhood of $\omega_0$.

  (\textit{Proof of item \ref{it:movsetmv}}) This statement follows from item \ref{it:movset} and the fact that
  for every $\omega\in (\RR^k)^{*}$ we have
  \[\mv(A_1(\omega),\ldots, A_{n+1}(\omega)) = \frac{1}{(n+1)!}\sum_{I\subseteq
      \irng{1,n+1}}(-1)^{n+1-|I|}\vol(\Conv(\sum_{i\in I}A_{i}(\omega))).\]
\end{proof}

Let now $A(\mathbf{w})$ be a $k$-variate moving set in $\RR^{n+1}$ and
let $r : 2^{A(\ww)} \rightarrow \NN$ be a matroid of rank $n + 1$ on
$A(\ww)$. Slightly abusing notation, we consider $r$ also as a matroid
on $A(\omega)$ for every $\omega\in (\RR^k)^{*}$. From \Cref{prop:movset} we now
obtain
\begin{corollary}
  \label{prop:mvpiecewise}
  Let $\omega_0\in (\RR^k)^{*}$. If there is a mixed subdivision
  $\mathcal{M}$ of $(r, A)$ so that the collection of tuples of subsets
  $\mathcal{M}(\mathbf{w})$ of $A(\ww)$ corresponding to $\mathcal{M}$ remains a mixed
  subdivision of $(r, A(\omega))$ for every $\omega$ in a small neighborhood of
  $\omega_0$, then the function $\omega \mapsto \mv(r, A(\omega))$ is linear in the same
  neighborhood of $\omega_0$. This condition holds for every
  $\omega_0$ in a dense open subset $U$ of $(\RR^k)^{*}$.
\end{corollary}
\begin{proof}
  Choose $\omega_0 \in (\RR^k)^{*}$ and let $\mathcal{M}$ be the set of dual tropical
  roots of $(r, A(\omega_0))$ at some arbitrary height vector
  $d\in \RR^{A(\omega_0)}$. By \Cref{thm:mv} we have
  \[\mv(r, A(\omega_0)) = \sum_{(S_1,\dots,S_r)\in \mathcal{M}} \mv(S_1,\dots,S_r).\]
  If $\omega_0$ lies inside a dense open subset of $(\RR^k)^{*}$ then the
  collection of tuples of subsets $\mathcal{M}(\mathbf{w})$ of
  $A(\ww)$ corresponding to $\mathcal{M}$ remains the set of dual tropical roots
  at height $d$ for every $\omega$ in a small neighborhood of
  $\omega_0$. This follows from the definition of mixed cells as
  cancellations of the lifted MCI $(r_d,A(\omega_0)_d)$. The statement
  of the corollary then follows from \Cref{prop:movset}.
\end{proof}

\begin{remark}
  Here, and in the following, we consider MCI's where the underlying
  support sets do not necessarily have integer entries. Of course,
  associated mixed volumes and tropicalizations are still well defined
  by choosing a suitable lattice other than $\ZZ^n$.
\end{remark}

\begin{notation}
  If $M$ is a mixed cell of the MCI $(r, A(\omega_0))$ for some
  $\omega_0\in (\RR^k)^{*}$, then the linear function
  $\omega \mapsto \vol(M(\omega))$, where $M(\ww)$ is the collection of subsets of
  $A(\ww)$ such that $M(\omega_0) = M$, will be denoted by $\svol(M)(\ww)$.
\end{notation}

\subsection{Computing Eliminant Polytopes of ECI's}
\label{sec:elimpol}

With notation as at the start of this section, we now outline a method
to compute the Newton polytope of $g$, starting with its support function: 
\begin{notation}
  Let $\Delta$ be the Newton polytope of $g$ and let $\melim : (\RR^k)^{*} \rightarrow \RR$ be its
  support function.
\end{notation}
Note again that $\Delta$, and hence also $\melim$, only depends on the
input data $(V, A)$.

Out starting point is the formula for $\melim$ given in Theorem 2.14
of \cite{esterovEngineeredCompleteIntersections2025}. We give
here an effective reformulation of this result.

To this end, let $(r,A)$ be the MCI associated to the ECI
$(V,A)$.
\begin{notation}
  Denote by $\pi:\RR^n\times \RR^k\rightarrow \RR^n$ and by
  $p : \RR^n\times \RR^k \rightarrow \RR^k$ the canonical projections.
  For symbolic variables $\ww = (w_1,\dots,w_k)$, define the moving set
  \[A(\ww) := \setdes{(\pi(a),\ww \cdot p(a))}{a\in A},\] and for any
  constant $C\in \ZZ$ define
  \[A_{C} := \setdes{(\pi(a),C)}{a\in A} \subset \ZZ^{n+1}.\]

  For arbitrary $\omega \in (\RR^k)^{*}$ and $C$, let
  $r_{\text{elim}}$ be the matroid defined on $A(\omega) \sqcup A_C$ via
  $r_{\text{elim}}(B) = r(B')$ where $B'$ is the canonical preimage of
  $B$ in $A$.
\end{notation}

Both $A(\ww)$ and $A_C$ are to be understood as multisets, i.e.
$A(\ww)$ and $A_C$ are both in canonical bijection with $A$.

With this notation we have
\begin{theorem}
  \label{thm:msv}
  For any $\omega\in (\RR^{k})^{*}$, let $C\in \ZZ$ be so that 
  \[C\leq \min\left(\set{0}\cup \setdes{\omega(a)-1}{a\in A}\right).\]
  Then
  \[\melim(\omega) = \mv(r_{\text{elim}},A(\omega) \sqcup A_C) - \mv(r_{\text{elim}},A(0) \sqcup A_C).\]
\end{theorem}
\begin{proof}
  Let $m_1,\dots,m_{n+1}$ be the conewise linear functions associated
  to $(r,A)$. According to Theorem 2.14 in
  \cite{esterovEngineeredCompleteIntersections2025} we have the
  equality
  \[\melim(\omega) = \frac{1}{(n+1)!}\delta^{n+1}(m_{\omega}),\]
  where $m_{\omega}(\gamma,t)$ is defined as
  $\prod_im_i(\gamma,t\omega) - \prod_im_i(\gamma,0)$ for $t\geq 0$ and as
  $0$ for $t < 0$. Here, the notation $\delta^{n+1}(m_{\omega})$ refers to the
  iterated corner locus of the conewise \textit{polynomial} function
  $m_{\omega}$, we refer to \cite{esterovTropicalVarietiesPolynomial2012}
  for details and remark here only that if a conewise polynomial
  function $\tilde{m}$ is the product of piecewise linear functions
  $\tilde{m}_1,\ldots,\tilde{m}_{n+1}$, then
  \[\frac{1}{(n+1)!}\delta^{n+1}(\tilde{m}) = \delta(\tilde{m}_1,\ldots,\tilde{m}_{n+1}).\]

  This now implies the desired statement: Indeed, for $C$ as chosen in
  the theorem, the function $m_{\omega}$ is the difference of the product
  of the conewise linear functions associated to
  $(r_{\text{elim}},A(\omega) \sqcup A_C)$ and the product of the conewise linear
  functions associated to $(r_{\text{elim}},A(0)\sqcup A_C)$. We conclude by noting
  that $\delta^{n+1}$ is a linear operator, with the sum of
  tropicalizations understood as in
  \cite{esterovTropicalVarietiesPolynomial2012}.
\end{proof}

While \Cref{thm:msv}, in combination with \Cref{alg:msd}, immediately
yields an algorithm to evaluate $\melim$ at arbitrary
$\omega\in (\RR^{k})^{*}$, our goal is to construct a {\em vertex oracle} for
$\Delta$, i.e. a way to compute, given an arbitrary generic
$\omega\in (\RR^{k})^{*}$, the vertex of $\Delta$ where $\omega$ is maximized. As
mentioned before, given such an oracle, one can then compute $\Delta$ using
the method given in \cite{hugginsIB4eSoftwareFramework2006}.

Such an oracle can now be provided using \Cref{prop:mvpiecewise}: 

\begin{corollary}
  \label{cor:vorac}
  Let $\omega\in (\RR^{k})^*$ and let $\mathcal{M}$ be a mixed subdivision
  of the MCI $(r_{\text{elim}},A(\omega)\sqcup A_C)$, with $C$ chosen as in
  \Cref{thm:msv}. Then the vertex of $\Delta$ at which $\omega$ is maximized is
  given by the coefficients the linear form
  \[\sum_{M\in \mathcal{M}}\svol(M)-\svol(M)(0).\]
\end{corollary}
\begin{proof}
  Let $\mathcal{M}'$ be a mixed subdivision of the MCI
  $(r_{\text{elim}},A_C\sqcup A_0)$.  We define the affine function
  \[m(\mathbf{w}) := \sum_{M\in \mathcal{M}}\svol(M) - \sum_{M'\in \mathcal{M}'}\vol(M').\]
  We can now find a basis $B$ of $(\RR^k)^{*}$ such that $\omega\in B$ and such that, according to \Cref{thm:msv} and \Cref{prop:mvpiecewise}, we have
  $\melim(b) = m(b)$ for every $b\in B$. But this must mean
  that the constant coefficient of $m$ is zero, i.e.
  \[\sum_{M\in \mathcal{M}}\svol(M)(0) = \sum_{M'\in \mathcal{M}'}\vol(M'),\]
  which implies the desired expression.
\end{proof}

\Cref{cor:vorac} can now be coupled with \Cref{alg:thc} in order to
obtain an algorithm for computing vertices of $\Delta$. Assume that we have
computed the set $\mathcal{M}_1$ of dual tropical roots of the MCI
$(r_{\text{elim}},A(\omega_1)\sqcup A_C)$ with
$\omega_1\in (\RR^{k})^{*}$ at some height
$d_1\in \RR^{A(\omega_1)\sqcup A_C}$. We now want to obtain the set of dual
tropical roots $\mathcal{M}_2$ of the MCI
$(r_{\text{elim}},A(\omega_2)\sqcup A_C)$ for another
$\omega_2\in (\RR^{k})^{*}$ at an arbitrary height vector
$d_2\in \RR^{A(\omega_1)\sqcup A_C}$.

We may assume that we can choose $C = 0$ in both cases and that the
$\omega_i$ have only positive entries, otherwise we shift
$A(\omega_i)\sqcup A_C$ by the vector $(0,-C)\in \ZZ^{n+1}$. For a suitably chosen
$\lambda\in \NN$, the convex hull of $A(\omega_2)\sqcup A_0$ is now contained in the
convex hull of $A(\lambda \omega_1)\sqcup A_0$ and $\mathcal{M}$ gives a mixed subdivision of
$A(\lambda \omega_1)\sqcup A_0$ by identifying $A(\omega_1)$ with
$A(\lambda \omega_1)$. This means now that the desired subdivision
$\mathcal{M}_2$ can be computed using \Cref{lem:startsystem}. The results of
this section together with this observation are summarized in
\Cref{alg:suppfunc}.

\begin{algorithm}
  \caption{Computing a vertex of $\Delta$}
  \label{alg:suppfunc}
  \begin{description}
  \item[Input] An ECI $(V,A)$ as above, two covectors
    $\omega_1,\omega_2\in (\RR_{>0}^{k})^{*}$, the dual tropical roots
    $\mathcal{M}_1$ of $(r_{\text{elim}},A(\omega_1)\sqcup A_0)$ at height
    $d_1\in \RR^{A_{\omega_1}\sqcup A_0}$.
  \item[Output] The vertex of $\Delta$ at which
    $\omega_2$ attains its maximum.
  \end{description}
  \begin{pseudo}
    \kw{function} \fn{VertOrac}((V,A),\omega_1, \omega_2, \mathcal{M}_1, d_1)\\+
    $d_2 \gets $ \tn{any generic height vector in} $\RR^{A(\omega_2)\sqcup A_0}$\\
    $\lambda\gets $ \tn{any number in} $\NN$ \tn{so that} $\Conv(A(\omega_2)\sqcup A_0)\subseteq \Conv(A(\lambda \omega_1)\sqcup A_0)$\\
    $r''\gets$ \tn{the matroid on} $A(\omega_2)\sqcup A_0\sqcup A(\lambda \omega_1)\sqcup A_0$ \tn{constructed from} $r_{\text{elim}}$ \tn{as in \Cref{lem:startsystem}}\\
    $\mathcal{M}'\gets \fn{Homotopy}((r'',A(\omega_2)\sqcup A_0\sqcup A(\lambda \omega_1)\sqcup A_0),\mathbf{1}_{A(\lambda \omega_1)\sqcup A_0}+\epsilon d_1,\mathcal{M}_1,\mathbf{1}_{A(\omega_2)} + \epsilon d_2)$\\
    $\mathcal{M}_2 \gets \setdes{(S_1,\dots,S_k)\in \mathcal{M}'}{S_i\subset A(\omega_2)\sqcup A_0 \; \forall i \in \irng{1,k}}$\\
    return $\sum_{M\in \mathcal{M}_2}\svol(M)-\svol(M)(0)$\\
  \end{pseudo}
\end{algorithm}

\begin{remark}
  \label{rem:allverts}
  We observed in practice that using \Cref{alg:suppfunc} to compute
  vertices of $\Delta$ is in general faster than using \Cref{alg:msd} to
  compute the required mixed subdivisions. Furthermore, as mentioned
  before, \Cref{alg:suppfunc} can be combined with the technique given
  in \cite{hugginsIB4eSoftwareFramework2006} to compute all vertices
  of the polytope $\Delta$. The same technique was used e.g. in the
  computation of eliminant polytopes in
  \cite{mohrComputationNewtonPolytopes2025} and
  \cite{roseTropicalImplicitizationRevisited2025}.
\end{remark}

\section{Examples}
\label{sec:examples}

A software implementation of \Cref{alg:thc}, \Cref{alg:msd} and for
computing Newton polytopes of eliminants of ECI's based on
\Cref{alg:suppfunc} and \cite{hugginsIB4eSoftwareFramework2006} is
available at
\begin{center}
  \url{https://github.com/RafaelDavidMohr/MCISubdivisions.jl}.
\end{center}
This implementation is written in the programming language
\texttt{Julia} \cite{bezanson2017} and uses functionality of the
computer algebra system \texttt{Oscar} \cite{OSCAR-book}. Below we
gather several examples to which we have applied this implementation.
All computations detailed below were performed on a single core of an
Intel i7-8665U @ 4.8Ghz CPU, using version 1.12.6 of \texttt{Julia}
and version 1.7.2 of \texttt{Oscar}.

\paragraph*{Chemical Reaction Networks.}

Steady state equations of dynamical systems associated to certain {\em
  chemical reaction networks} are systems of polynomial
equations of the shape
\[F =
  \begin{cases}
    V' \cdot (c_a \cdot \xx^a\;|\; a\in A')^T &= 0\\
    L \cdot \xx - \mathbf{d} &=0,
  \end{cases}
\]
where $A'\subset \ZZ^n$ is a finite multiset, $V'\in \RR^{s \times A}$ is of full row
rank, $c_a$, $a\in A'$, are symbolic parameters, $\xx$ consists of $n$
variables, $L\in \RR^{(n - s) \times n}$ and $\mathbf{d}$ consists of
$n-s$ symbolic parameters. We refer e.g. to
\cite{conradiIdentifyingParameterRegions2017} for a short introduction
on how to construct the equations $F$ from a given chemical reaction
network.

Note that $F$ is almost an ECI, only symbolic parameters in the linear
part of $F$, specified by $L$, are missing.  Such a system $F$ was
referred to as an {\em augmented vertically parametrized system} in
\cite{feliuRootBoundsVertical2026}. In order to obtain an ECI from $F$
on which we can test \Cref{alg:msd}, we simply introduced the missing
parameters in the linear part of $F$ and let $(V, A)$ be the resulting
ECI.

\Cref{tab:reaction} records the results of running our implementation
of \Cref{alg:msd} on a few examples, a selection of which were sourced
from the database \url{odebase.org}
\cite{ludersODEbaseRepositoryODE2022}. The examples labeled
``$k$-site'' in \Cref{tab:reaction}, for some $k$, are described in
\cite{feliuRootBoundsVertical2026}. We record the mixed volume
$\mv(V,A)$ computed by \Cref{alg:thc} as well as the normalized volume
of the convex hull of $A$, which is the root count of a system with
generic coefficients and support $A$, by Kouchnirenko's theorem
\cite{kouchnirenko1976}.

\begin{table}[H]
  \small
\centering
\caption{Results of running \Cref{alg:msd} on ECI's coming from chemical reaction networks}
\begin{tabular}{ |c|c|c|c|c|c| } 
  \hline
  Source of example & $|\xx|$ & $|A|$ & $\vol(A)$ & mixed volume & runtime\\
  \hline\hline
  \href{https://odebase.org/detail/1603}{odebase: example 1} & 13 & 23 & 48 & 25 & 0.2s\\
  \hline
  \href{https://odebase.org/detail/1532}{odebase: example 2} & 19 & 35 & 164 & 45 & 2.7s\\
  \hline
  \href{https://odebase.org/detail/1591}{odebase: example 3} & 27 & 44 & 307 & 166 & 14.1s\\
  \hline
  \href{https://odebase.org/detail/1641}{odebase: example 4} & 35 & 51 & 16 & 16 & 0.8s\\
  \hline
  \cite{grossAlgebraicSystemsBiology2016} & 19 & 26 & 60 & 9 & 0.7s\\ 
  \hline
  $9$-site & 30 & 49 & 64 & 19 & 3.1s\\
  \hline
  $11$-site & 36 & 59 & 89 & 23 & 7.0s\\
  \hline
  $13$-site & 42 & 69 & 118 & 27 & 16.9s\\
  \hline
\end{tabular}
\label{tab:reaction}
\end{table}

The database \url{odebase.org} contains five examples with more than
35 variables. Two of these yield an ECI consisting of linear
equations. The three remaining examples involve respectively 86, 90
and 194 variables. We attempted to run our implementation of
\Cref{alg:msd} on the \href{https://odebase.org/detail/1334}{example}
involving 86 variables: Here already the computation of the regular
triangulation of the convex hull of the underlying support set $A$
needed in \Cref{alg:msd} did not terminate within 2 hours of
computation.

Nonetheless, we note the following improvement compared to the state
of the art: The technique described in
\cite{helminckTropicalMethodSolving2024} for computing mixed volumes
of ECI's was therein reported to take roughly 60 seconds of
computation time on the example coming from
\cite{grossAlgebraicSystemsBiology2016}. The technique described in
\cite{feliuRootBoundsVertical2026} to tropicalize augmented vertically
parametrized systems was therein reported to take roughy 262 seconds
to to compute the mixed volume of the $9$-site reaction network. Both
of these examples are solved by our implementation in under 5 seconds.

A script to repeat the computations necessary to produce
\Cref{tab:reaction} can be found in the \texttt{examples} folder of
the code repository linked above.

\paragraph*{Real Patchworking.}

We next give an example illustrating the real patchworking techniques
discussed in \Cref{sec:realpatch}. Combining those techniques with
\Cref{alg:thc} and \Cref{alg:msd} allowed us to prove the following
result computationally, see also \Cref{fig:cusps}:

\begin{theorem}
  \label{thm:disc}
  There exists a polynomial $f\in \RR[x,y,z]$ of degree 4 such that all
  24 cusp singularities of the discriminant curve of $f$ lie in $\RR^3$, i.e.
  \[\Var(f,\partial_xf,\partial_x^2f) \subset \RR^3.\]
\end{theorem}
\begin{proof}
  Let $A\subset \ZZ^3$ be the set of exponent vectors of all monomials in the
  variables $x,y,z$ of degree at most $4$. Let $\tilde{f}$ be the
  polynomial defined in Appendix \ref{sec:datadisc}. Let $(V,A)$ be the ECI
  defined by letting the rows of $V$ be the coefficients of the polynomials
  $\tilde{f}$, $x\partial_x\tilde{f}$ and $x^2\partial_x^2\tilde{f}$ and let $(r,A)$ be the
  corresponding MCI.

  Additionally, with the quantities
  $d_{r}$, $d_{\varepsilon}$, as defined in
  Appendix \ref{sec:datadisc}, define
  \[d := d_{r} + \varepsilon d_{\varepsilon} \in \mathbb{D}^A.\] Here
  $\mathbb{D} := \RR[\varepsilon]/\langle \varepsilon^2\rangle$ denotes the ring of dual numbers.  We
  used our implementation of \Cref{alg:thc} to compute the set
  $\mathcal{M}$ of dual tropical roots of $(r,A)$ at height $d$. As explained in
  \Cref{rem:sympert}, \Cref{alg:thc} can compute with height vectors
  whose entries lies in $\mathbb{D}$. The set $\mathcal{M}$ consists of
  $12$ individual cells, each containing either three tuples of pairs
  of elements of $A$ or one tuple and one triple of elements of
  $A$. Next we used \Cref{thm:rootcounting} to compute for each
  $M\in \mathcal{M}$ the number of real solutions of the polynomial system given
  by the cancellation of $(r,A)$ corresponding to $M$ introduced
  above. We further confirmed these real root counts by solving these
  systems using the software \texttt{msolve}
  \cite{berthomieu2021}. The sum of these numbers equals $24$ which
  matches the mixed volume of $(V,A)$. Thus the claim follows from
  \Cref{thm:realrootcount}.
\end{proof}

A script to repeat the computations necessary for the proof of
\Cref{thm:disc} can be found in the \texttt{examples} folder of the
code repository linked above.

Let us explain how we arrived at the quantities defined in Appendix
\ref{sec:datadisc}. The polynomial $\tilde{f}$ was generated as
\[\tilde{f} := \sum_{a + b + c \leq 4}s_{a,b,c}x^ay^bz^d\]
with randomly chosen $s_{a,b,c}\in \set{\pm 1}$. We then computed the set
$\mathcal{M}$ of dual tropical roots of the MCI $(r,A)$ defined in the proof of
\Cref{thm:disc} at randomly generated height vectors
$d\in \mathbb{D}^A$ using \Cref{alg:msd} and \Cref{alg:thc}, checking
each time if the sum of the number of real of roots of the
cancellations of $(r, A)$ at $\mathcal{M}$ is $24$ and stopping the sampling of
random height vectors $d$ if this number is reached. Computing a
single such mixed subdivision $\mathcal{M}$ took a fraction of a second using
our implementation. However, height vectors $d$ producing with an
associated number of real roots equal to 24 seem to be ``rare'': We
needed to sample roughly 13000 height vectors before finding the
height vector $d$ used in the proof of \Cref{thm:disc}.

\paragraph*{$A$-discriminants.}

So-called {\em $A$-discriminants} \cite{gelfand1994} can be computed
as eliminants of ECI's in the setting introduced at the beginning of
\Cref{sec:dte}. Given a set $A\subset \ZZ^n$, the associated $A$-discriminant
is constructed as follows: Let $\xx$ be a set of $n$ variables and introduce
one variable $z_a$ for each $a\in A$. Further define
\[f_{A}(\xx) := \sum_{a\in A} z_a\xx^a.\] The {\em $A$-discriminant
  variety} $\nabla_A$ is then defined as the Zariski closure of
\[\nabla_A^{\circ} := \setdes{z\in \PP^{n-1}}{\Var(f_{A})\text{ has a singular
    point}}.\] When $\nabla_A$ has codimension one, there is an irreducible
polynomial $\Delta_A\in \ZZ[z_a\;|\;a\in A]$ with $\nabla_A = \Var(\Delta_A)$, this polynomial
$\Delta_{A}$ is the {\em $A$-discriminant.}

To compute Newton polytopes of $A$-discriminants using
\Cref{alg:suppfunc} combined with the algorithm given in
\cite{hugginsIB4eSoftwareFramework2006}, as explained in
\Cref{rem:allverts}, we can proceed as follows: Choose one $a_0\in A$,
introduce a variable $z_a$ for each $a\in A\setminus \set{a_0}$ and define
\[\tilde{f_A}(\xx) := \xx^{a_0} + \sum_{a\in A\setminus \set{a_0}} z_a\xx^a.\]
We then let $V$ be the coefficient matrix of the polynomial system
\[\tilde{f_A} = x_1\partial_{x_1}\tilde{f_A} = \ldots = x_n\partial_{x_n}\tilde{f_A}\]
and let $\tilde{A}$ be the set of exponent vectors appearing in this
system.  Note that $V$ is the matrix whose columns correspond to $A$
with a row of ones added at the top, when this matrix has maximal
rank, we obtain an ECI $(V, \tilde{A})$ to which we can apply
\Cref{alg:suppfunc} in combination with
\cite{hugginsIB4eSoftwareFramework2006}. If $\nabla_A$ has codimension one,
then we obtain the Newton polytope of $\Delta_A$, dehomogenized along the
variable $z_{a_0}$.

In \Cref{tab:adisc} we record the results of running this computation
for several different choices of $A$. The sets $A_1$, $A_2$ and $A_3$
mentioned in this table are given in Appendix \ref{sec:adiscdata}, the
last two rows of \Cref{tab:adisc} are indexed by products of
simplices, in this case $A$ corresponds to the set of vertices of the
corresponding product of simplices and the corresponding $A$-discriminant
is a so-called {\em hyperdeterminant} associated to a tensor.

We further record in \Cref{tab:adisc} the timing of using the
techniques described in
\cite{roseTropicalImplicitizationRevisited2025} to compute
tropicalizations of $A$-discriminants. These techniques are
implemented in the \texttt{Julia}-package
\href{https://github.com/kemalrose/TropicalImplicitization.jl}{\texttt{TropicalImplicitization.jl}},
which we ran on the examples recorded in \Cref{tab:adisc}. Note that
\cite{roseTropicalImplicitizationRevisited2025} describes a technique
to obtain a vertex oracle for the Newton polytope of a polynomial
given its tropicalization which can be combined with
\cite{hugginsIB4eSoftwareFramework2006} to compute this Newton
polytope. We record in \Cref{tab:adisc} only the time it took to
compute the tropicalization of the corresponding $A$-discriminant
using \cite{roseTropicalImplicitizationRevisited2025}.

We also attempted to compute the $A$-discriminant where $A$ is the set
of vertices of $\Delta_1 \times \Delta_2 \times \Delta_2$. While running \Cref{alg:suppfunc} a single
time in this instance finishes in a fraction of a second, our implementation
of \cite{hugginsIB4eSoftwareFramework2006} did not finish within 3 hours of
computation, suggesting that the underlying Newton polytope of $\Delta_A$ has a
very large number of vertices.

\begin{table}[H]
  \small
\centering
\caption{Results of running \Cref{alg:suppfunc} to compute Newton polytopes $A$-discriminants}
\begin{tabular}{ |c|c|c|c|c|c| } 
  \hline
   & \Cref{alg:suppfunc} & $\trop(\Delta_A)$ with \cite{roseTropicalImplicitizationRevisited2025} & dimension & \# vertices & \# lattice points\\
  \hline\hline
  $A_1$ & 2.3s & 1.2s & 7 & 45 & 43400 \\
  \hline
  $A_2$ & 5.4s & 3.5s & 9 & 64 & 43329 \\
  \hline
  $A_3$ & 237.5s & 69.1s & 11 & 416 & 7223\\
  \hline
  $\Delta_1\times \Delta_1 \times \Delta_1$ & 0.3s & 1.6s & 7 & 6 & 12\\
  \hline
  $\Delta_1\times \Delta_1\times \Delta_2$ & 7.5s & 65.4s & 11 & 60 & 66\\
  \hline
\end{tabular}
\label{tab:adisc}
\end{table}

\paragraph*{Euclidean Distance Degree.}
\label{edd}

The {\em euclidean distance degree}
\cite{draismaEuclideanDistanceDegree2016} of an algebraic variety $X$
encodes the number of critical points of the squared euclidean
distance function to a generic point outside of $X$. Let
$\CC[\xx^{\pm}]$ be an $n$-variate Laurent polynomial ring and suppose
that $X\subset (\CC^{*})^{n}$ is smooth and cut out by a reduced regular
sequence $f_1,\dots,f_c\in \CC[\xx^{\pm}]$. Introducing a vector of
$n$ new variables $\mathbf{u}$ and $c$ new variables $z_1,\dots,z_c$,
the euclidean distance degree is the number of solutions of the system
\begin{equation}
  \label{eq:edd}
  f_1 = \ldots = f_c = (\mathbf{u} - \xx) - \sum_{i=1}^c z_i \nabla f_i = 0 
\end{equation}
for generic values of $\mathbf{u}$, where $\nabla f_i$ denotes the gradient
of $f_i$.

If the $f_i$ are generic polynomials with fixed Newton polytopes then,
after after replacing the derivatives $\partial_{x} f_i$ for
$x\in \xx$ by the corresponding toric derivatives $x\partial_xf_i$ by
multiplying by $x$, \eqref{eq:edd} is an ECI.

In order to further evaluate the performance of \Cref{alg:suppfunc}
(and thus also of \Cref{alg:msd}) we generated instances of
\eqref{eq:edd} by first choosing the number $n$ of the variables
$\xx$, by then letting $\tilde{A}\subset \NN^n$ be a set of $k$ random
distinct exponent vectors with entries between $0$ and $4$, together
with the vector containing only zeros, and by finally instantiating
equations $f_1,\dots,f_{n-1}$ with random coefficients and support
$\tilde{A}$. We denote such an instance by $[n,k]$ in \Cref{tab:edd}.

In \Cref{tab:edd} we record timings for the following computation: For
each considered instance of \eqref{eq:edd}, we used
\Cref{alg:suppfunc} to compute a single vertex of the Newton polytope
of the eliminant resulting from \eqref{eq:edd} after eliminating the
variables $z_1,\ldots, z_{n-1}$ and all but one of the variables in
$\xx$. This timing is recorded in the fourth column of
\Cref{tab:edd}. Such an elimination computation can be potentially
used to compute the so-called {\em ED discriminant}, see Example 7.2
in \cite{draismaEuclideanDistanceDegree2016}. Note that
\Cref{alg:suppfunc} already requires a mixed subdivision coming from
an application of \Cref{cor:vorac} as an input, which is then tracked
using a homotopy to a mixed subdivision of the target MCI. We computed
such an input mixed subdivision using \Cref{alg:msd}. The time for
this computation is recorded in the fifth column of \Cref{tab:edd}.

\begin{table}[H]
  \small
\centering
\caption{Results of running \Cref{alg:suppfunc} on instances of \eqref{eq:edd}}
\begin{tabular}{ |l|c|c|c|c| } 
  \hline
   & $|A|$ & \# elim. variables & Direct usage of \Cref{cor:vorac} & One run of \Cref{alg:suppfunc}\\
  \hline\hline
  $[3, 9]$ & 34 & 6 & 3.0s & 1.3s\\
  \hline
  $[3, 13]$ & 46 & 6 & 5.7s & 3.3s\\
  \hline
  $[3, 17]$ & 57 & 6 & 14.1s & 5.5s\\
  \hline
  $[3, 21]$ & 70 & 6 & 15s & 5.7s\\
  \hline
  $[4, 5]$ & 29 & 8 & 14.3s & rounding error\\ 
  \hline
  $[4, 7]$ & 37 & 8 & 53.6s & rounding error\\
  \hline
  $[4, 9]$ & 45 & 8 & 87.8s & rounding error\\
  \hline
\end{tabular}
\label{tab:edd}
\end{table}

As shown in \Cref{tab:edd}, the difficulty of computing a
single vertex of an eliminant Newton polytope of ECI's using a direct
application of \Cref{cor:vorac} or \Cref{alg:suppfunc} depends on the
support size of the underlying ECI as well as on the number of
eliminated variables, both of which are recorded in \Cref{tab:edd}. \Cref{tab:edd}
also shows that using \Cref{alg:suppfunc} is generally faster than a direct application
of \Cref{cor:vorac}.

Let us explain the rounding errors recorded in \Cref{tab:edd}:
\Cref{alg:suppfunc} typically manipulates support sets in which one
entry (the one depending on the covector maximized at the vertex to be
computed) in every element is very large compared to the remaining
entries. For efficiency, our implementation of \Cref{alg:msd} relies
on floating point computations and subsequent rounding in order to
compute defining facet inequalities of mixed cell cones. The large
difference between the size of the entries of the support sets
considered in \Cref{alg:suppfunc}, however, tends to produce very
poorly conditioned linear systems of which the facet inequalities are
solutions. This phenomenon produces the numerical rounding errors
recorded in \Cref{tab:edd}. Of course one could improve our software
by falling back to an exact computation of the facet inequalities in
the case that a rounding error is detected.

\newpage

\appendix

\section{Data used in the proof of \Cref{thm:disc}}
\label{sec:datadisc}

\begin{align*}
  \tilde{f} := &-x^4 - x^3y + x^3z - x^3 - x^2y^2 - x^2yz - x^2y - x^2z^2 - x^2z\\
            &- x^2 - xy^3 - xy^2z + xy^2 - xyz^2 - xyz + xy + xz^3 + xz^2 - xz - x + y^4 \\
  &+ y^3z + y^3 - y^2z^2 + y^2z - y^2 - yz^3 - yz^2 + yz - y + z^4 + z^3 - z^2 - z - 1
\end{align*}

\vspace{-0.6cm}
\begin{align*}
  d_{r} := (&-40, -45, -30, -60, 40, 20, -108, -12, -90, -20, -105, 105, -40, 60, -30,\\
            &15, -20, -12, 12, -80, -40, -140, 45, 20, 0, 20, -20, -105, -12, 100,\\
            &-15, 40, -40, -20, -45)\\
  d_{\varepsilon} := (&-46629, -9066, -18529, 28514, 31055, -22373, 569, 24243, -45323,\\
            &-25435, -3540, -4411, 6774, 44154, 10846, -37613, 17184, -45248,\\
            &-26873, -4428, -27715, 20531, 36543, 24816, 31217, 15799, -43881,\\
            &-39483, -23808, 22715, -29401, -31858, -46102, -5886, -42199)
\end{align*}

\section{Data used to produce \Cref{tab:adisc}}
\label{sec:adiscdata}

\[
  A_1 =
  \begin{bmatrix}
    0 & 0 & 0 & 0 & 1 & 1 & 1 & 1 \\
    2 & 3 & 5 & 7 & 11 & 13 & 17 & 19 \\
    19 & 17 & 13 & 11 & 7 & 5 & 3 & 2
  \end{bmatrix}
\]

\[
  A_2 =
  \begin{bmatrix}
    0 & 0 & 0 & 0 & 0 & 1 & 1 & 1 & 1 & 1 \\
    2 & 3 & 5 & 7 & 11 & 13 & 17 & 19 & 21 & 23 \\
    23 & 21 & 19 & 17 & 13 & 11 & 7 & 5 & 3 & 2
  \end{bmatrix}
\]

\setcounter{MaxMatrixCols}{12}
\[
A_3 =
  \begin{bmatrix}
    1 & 0 & 1 & 1 & 1 & 0 & 1 & 0 & 1 & 0 & 1 & 1 \\
    1 & 1 & 1 & 1 & 0 & 0 & 0 & 1 & 0 & 1 & 1 & 0 \\
    1 & 0 & 1 & 0 & 0 & 0 & 1 & 0 & 0 & 1 & 0 & 1 \\
    1 & 0 & 0 & 0 & 1 & 0 & 1 & 1 & 0 & 1 & 1 & 0
\end{bmatrix}
\]

\newpage
\printbibliography
\end{document}